\documentclass{new_tlp}

\usepackage[semibold,scale=0.875]{sourcecodepro}
\usepackage[utf8]{inputenc}
\usepackage{amsmath}
\usepackage{amssymb}
\usepackage{microtype}
\usepackage{url}
\usepackage{xcolor}
\usepackage{listings}
\let\proof\relax

\usepackage{amsthm}

\newtheorem{definition}{Definition}
\newtheorem{lemma}{Lemma}
\newtheorem{proposition}{Proposition}

\newtheorem{theorem}{Theorem}

\DeclareMathOperator{\myinf}{\mathit{inf}}

\DeclareMathOperator{\mysup}{\mathit{sup}}
\DeclareMathOperator{\body}{\mathit{Body}}
\DeclareMathOperator{\head}{\mathit{Head}}
\DeclareMathOperator{\mynot}{\mathit{not}}
\DeclareMathOperator{\valop}{\mathit{val}}
\newcommand{\val}[2]{\valop_{#1}(#2)}
\DeclareMathOperator{\isint}{\mathit{is\_int}}

\newcommand{\I}{\mathcal{I}}
\newcommand{\J}{\mathcal{J}}
\newcommand{\Iu}{\mathcal{I}^{\uparrow}}
\newcommand{\Ju}{\mathcal{J}^{\uparrow}}

\newcommand{\T}{\mathcal{T}}
\newcommand{\Num}{\mathcal{Z}}

\newcommand{\numeral}[1]{\overline{#1}}
\newcommand{\op}{\circ}

\newcommand{\anthem}{\textsc{anthem}}

\newcommand{\gringo}{\textsc{gringo}}
\newcommand{\smodels}{\textsc{smodels}}
\newcommand{\vampire}{\textsc{vampire}}

\newcommand{\boldh}{\mathbf{h}}
\newcommand{\boldH}{\mathbf{H}}
\newcommand{\boldi}{\mathbf{i}}
\newcommand{\boldp}{\mathbf{p}}
\newcommand{\boldP}{\mathbf{P}}
\newcommand{\boldq}{\mathbf{q}}
\newcommand{\boldt}{\mathbf{t}}
\newcommand{\boldu}{\mathbf{u}}
\newcommand{\boldU}{\mathbf{U}}
\newcommand{\boldv}{\mathbf{v}}

\newcommand{\twodots}{\mathinner {\ldotp \ldotp}}

\newcommand{\eqdef}{\mathrel{\vbox{\offinterlineskip\ialign {\hfil##\hfil\cr $\scriptscriptstyle\mathrm{def}$\cr \noalign{\kern1pt}$=$\cr \noalign{\kern-0.1pt}}}}}

\newcommand{\Cdf}[1]{\ensuremath{\mathit{Cdef}(#1)}}
\newcommand{\Clark}[1]{\ensuremath{\mathit{Clark}(#1)}}
\newcommand{\SM}{{\rm SM}}
\newcommand{\COMP}{{\rm COMP}}
\newcommand{\grp}[3]{\mathit{gr}^{#1}_{#2}(#3)}
\newcommand{\gr}[2]{\grp{\boldp}{#1}{#2}}
\newcommand{\prop}[1]{{#1}^{\mathtt{prop}}}

\newcommand{\SMstable}{\mbox{\ensuremath{\boldp}-stable}}
\newcommand{\INFstable}{\mbox{{\rm INF}-$\boldp$-stable}}
\newcommand{\Choice}[1]{\mathit{Choice}(#1)}
\newcommand{\At}[1]{#1^{\boldp}}

\newcommand{\unsorted}[1]{{#1}^{\mathtt{us}}}
\newcommand{\equivs}{\equiv_{\mathit{s}}}

\newcommand{\covered}{\mathit{covered}}
\newcommand{\Covered}{\mathit{Covered}}

\newcommand{\incover}{\mathit{in\_cover}}
\newcommand{\PH}{\mathit{PH}}
\newcommand{\In}{\mathit{In}}
\newcommand{\Out}{\mathit{Out}}

\def\I{\mathcal{I}}
\def\ar{\leftarrow}
\def\rar{\rightarrow}
\def\lrar{\leftrightarrow}

\lstdefinelanguage{anthem}
{
	keywords={and, exists, forall, integer, not, or},
	otherkeywords={assume, axiom, backward, forward, input, lemma, output, spec},
	morekeywords={},
	morecomment=[l]{\#},
	sensitive,
}

\lstdefinelanguage{clingo}
{
	keywords={not},
	otherkeywords={},
	morekeywords={},
	morecomment=[l]{\%},
	sensitive,
}

\definecolor{gray}{rgb}{0.33,0.33,0.33}

\theoremstyle{plain}
\newtheorem*{mainlemma}{Main Lemma}
\newtheorem*{mainlemmaio}{Main Lemma for IO-Programs}

\begin{document}

\title[Verifying Tight Logic Programs with \anthem{} and \vampire{}]{Verifying Tight Logic Programs\\with \anthem{} and \vampire{}}
\author[Jorge Fandinno, Vladimir Lifschitz, Patrick Lühne, and Torsten Schaub]{JORGE FANDINNO\\
	University of Potsdam, Germany
	\and
	VLADIMIR LIFSCHITZ\\
	University of Texas at Austin, USA
	\and
	PATRICK LÜHNE
	and
	TORSTEN SCHAUB\\
	University of Potsdam, Germany
}

\maketitle

\begin{abstract}
This paper continues the line of research aimed at investigating the relationship between logic programs and first-order theories.
We extend the definition of program completion to programs with input and output in a subset of the input language of the ASP grounder \gringo, study the relationship between stable models and completion in this context, and describe preliminary experiments with the use of two software tools, \anthem{} and \vampire, for verifying the correctness of programs with input and output.
Proofs of theorems are based on a lemma that relates the semantics of programs studied in this paper to stable models of first-order formulas.
This paper is under consideration for acceptance in \emph{Theory and Practice of Logic Programming}.
\end{abstract}
 \section{Introduction}
\label{sec:introduction}

This paper continues the line of research aimed at investigating the relationship between logic programs and first-order theories, which goes back to Keith Clark's classical paper on program completion \cite{clark78a}.
We are interested in extending that idea to input languages of answer set programming (ASP) \cite{breitr11a} and in using it for the verification of ASP programs.

These programs are different from logic programs studied by Clark in two ways.
First, their semantics is based on the concept of a stable model \cite{gellif88b}, which is not equivalent to the completion semantics, unless some assumptions about the program are made, such as tightness \cite{fages94a,erdlif03a}.
Second, ASP programs can include symbols for integers and arithmetic operations.
Ways to extend completion to these constructs are proposed in two recent publications \cite{halira17a,lilusc20a}.

It is essential for our purposes that, in ASP, we often deal with programs that accept an input and employ auxiliary predicates.
Consider, for instance, the program
\begin{equation}
	\label{ex1}
	\begin{array}{l}
		\{\incover(1 \twodots n)\},\\
		\ar I \neq J \land \incover(I) \land \incover(J) \land s(X, I) \land s(X, J),\\
		\covered(X) \ar \incover(I) \land s(X, I),\\
		\ar s(X, I) \land \mynot \covered(X),\\
	\end{array}
\end{equation}
which encodes exact covers\footnote
{
	\label{ft1}
	An \emph{exact cover} of a collection~$S$ of sets is a subcollection~$S'$ of~$S$ such that each element of the union of all sets in~$S$ belongs to exactly one set in~$S'$.
}
of a finite family of finite sets.
It accepts an input, represented by a placeholder~$n$ and a predicate symbol~$s/2$; $n$ is the number of sets, and the formula $s(X, I)$ expresses that~$X$ belongs to the $I$-th set.
The predicate symbol $\incover/1$ represents the output, and $\covered/1$ is auxiliary.

In this paper, we extend the definition of program completion to programs with input and output in a subset of the input language of the ASP grounder \gringo{}, state a theorem describing the relationship between stable models and completion in this context, and describe preliminary experiments with the use of two software tools, \anthem{} \cite{lilusc18a,lilusc19a} and \vampire{} \cite{kovvor13a}, for verifying the correctness of programs with input and output.
A lemma that may be of independent interest relates the semantics of simple \gringo{} programs to stable models of first-order formulas \cite{feleli11a}.
It shows, in other words, that the approach to the semantics of logic programs advocated by \citeN{lelipa08b} is applicable to the subset of \gringo{} studied here.
 \section{Review: Programs}
\label{sec:programs}

The class of programs studied in this paper is the subset of the input language of \gringo{} \cite{gehakalisc15a} discussed in earlier publications by Lifschitz~et~al.~\citeyear{lilusc19a,lilusc20a}.

We assume that three countably infinite sets of symbols are selected: \emph{numerals}, \emph{symbolic constants}, and \emph{program
variables}.
We assume a 1-to-1 correspondence between numerals and integers; the numeral corresponding to an integer~$n$ is denoted by~$\numeral{n}$.

\emph{Program terms} are defined recursively:
\begin{itemize}
	\item Numerals, symbolic constants, program variables, and the symbols $\myinf$ and $\mysup$ are program terms;
	\item if $t_1$, $t_2$ are program terms and $\op$ is one of the \emph{operation names}
		\[
			+ \quad - \quad \times \quad / \quad \backslash \quad \twodots
		\]
		then $(t_1 \op t_2)$ is a program term.
\end{itemize}
If $t$~is a term, then $-t$~is shorthand for $\numeral{0} - t$.

A program term, or another syntactic expression, is \emph{ground} if it does not contain variables.
A ground expression is \emph{precomputed} if it does not contain operation names (in other words, if it is a numeral, a symbolic constant, or one of the symbols $\myinf$ or $\mysup$).

We assume that a total order on precomputed program terms is chosen such that
\begin{itemize}
	\item \strut $\myinf$ is its least element and $\mysup$ its greatest element,
	\item for any integers~$m$ and $n$, $\numeral{m} < \numeral{n}$ iff $m < n$, and
	\item for any integer~$n$ and any symbolic constant~$c$, $\numeral{n} < c$.
\end{itemize}

An \emph{atom} is an expression of the form~$p(\boldt)$, where $p$~is a symbolic constant and~$\boldt$ is a tuple of program terms.
A \emph{literal} is an atom possibly preceded by one or two occurrences of \emph{not}.
A \emph{comparison} is an expression of the form $(t_1 \prec t_2)$, where $t_1$, $t_2$ are program terms and $\prec$ is one of the \emph{comparison symbols}
\begin{align}
	\label{rel}
	= \quad \neq \quad < \quad > \quad \leq \quad \geq
\end{align}

A \emph{rule} is an expression of the form
\begin{align}
	\label{rule}
	\head \ar \body,
\end{align}
where
\begin{itemize}
	\item $\body$ is a conjunction (possibly empty) of literals and comparisons and
	\item $\head$ is either an atom (then, we say that \eqref{rule}~is a \emph{basic rule}) or an atom in braces (then, \eqref{rule}~is a \emph{choice rule}), or empty (then, \eqref{rule}~is a \emph{constraint}).
\end{itemize}

A \emph{program} is a finite set of rules.

For example, \eqref{ex1}~is a program in the sense of this definition, except that the expression $\{\incover(1 \twodots n)\}$, strictly speaking, should be written as $\{\incover(\numeral{1} \twodots n)\} \ar $.

The semantics of programs uses a translation~$\tau$, which turns a program into a set of propositional combinations of precomputed atoms
\cite[Section~3]{lilusc19a}.
For example, the result of applying~$\tau$ to the program
\begin{align}
	\label{ex2}
	\begin{split}
		&\{p(\numeral{1} \twodots \numeral{3})\},\\
		&q(X + \numeral{1}) \ar p(X)
	\end{split}
\end{align}
consists of the ground formulas
\begin{align*}
	\top &\rar (p(\numeral{1}) \lor \neg p(\numeral{1}))
		\land (p(\numeral{2}) \lor \neg p(\numeral{2}))
		\land (p(\numeral{3}) \lor \neg p(\numeral{3})),\\
	p(\numeral{n}) &\rar q(\numeral{n + 1}) \quad \textrm{for all integers $n$, and}\\
	p(r) &\rar \top \quad \textrm{for all precomputed terms~$r$ other than numerals.}
\end{align*}
Using the translation~$\tau$, we define the semantics of programs as follows:
A set of precomputed atoms is a \emph{stable model} of a program~$\Pi$ if it is a stable model (answer set) of $\tau \Pi$ in the sense of Ferraris \citeyear{ferraris05a}.

\section{Review: Representing Programs by First-Order Theories}
\label{sec:formulas}

The translation~$\tau^*$ \cite{lilusc19a} converts a program into a finite set of first-order sentences with variables of two sorts.
Besides program variables, which are meant to range over precomputed program terms, these sentences contain \emph{integer variables}, which are meant to range over numerals.
The second sort is a subsort of the first.

A~\emph{predicate symbol} is a pair~$p/n$, where $p$~is a symbolic constant and $n$~is a nonnegative integer.
About a program or another syntactic expression, we say that a predicate symbol~$p/n$ \emph{occurs} in it if it contains an atom of the form $p(t_1, \dots, t_n)$.

For any program~$\Pi$,~$\tau^* \Pi$ consists of formulas over the two-sorted signature~$\sigma_\Pi$ including
\begin{itemize}
	\item all precomputed terms as object constants; a precomputed constant is assigned the sort \emph{integer} iff it is a numeral;
	\item the symbols $+$, $-$, and $\times$ as binary function constants; their arguments and values have the sort \emph{integer};
	\item all predicate symbols occurring in~$\Pi$ and the comparison symbols other than equality as predicate constants.
\end{itemize}
An atomic formula $(p/n)(t_1, \dots, t_n)$ can be abbreviated as $p(t_1, \dots, t_n)$.
An atomic formula $\prec(t_1, t_2)$, where~$\prec$ is a comparison symbol, can be written as $t_1 \prec t_2$.
\emph{Formulas} are formed from atomic formulas using the propositional connectives
\[
	\bot \textrm{ (``false'')} \quad \land \quad \lor \quad \rar
\]
and the quantifiers $\forall$ and $\exists$ as usual in first-order languages with equality.
We use $\top$~as shorthand for $\bot \rar \bot$, $\neg F$~as shorthand for $F \rar \bot$, and $F \lrar G$ as shorthand for $(F \rar G) \land (G \rar F)$.

Prior to defining~$\tau^*$, we define, for every program term $t$, a formula $\val{t}{Z}$, where $Z$ is a program variable that does not
occur in~$t$.
That formula expresses, informally speaking, that $Z$ is one of the values of~$t$.
The definition is recursive:
\begin{itemize}
	\item if $t$ is a numeral, a symbolic constant, a program variable, $\myinf$, or $\mysup$, then $\val{t}{Z}$ is $Z = t$;
	\item if $t$ is $(t_1 \op t_2)$, where $\op$ is $+$, $-$, or $\times$, then $\val{t}{Z}$ is
		\[
			\exists I J \, (Z = I \op J \land \val{t_1}{I} \land \val{t_2}{J})
		\]
		where $I$, $J$ are fresh integer variables;
	\item if $t$ is $(t_1 \, / \, t_2)$, then $\val{t}{Z}$ is
		\begin{align*}
			\exists I J Q R \, (I = J \times Q + R &\land \val{t_1}{I} \land \val{t_2}{J}\\
				&\land J \neq 0 \land R \geq 0 \land R < Q \land Z = Q),
		\end{align*}
		where $I$, $J$, $Q$, $R$ are fresh integer variables;
	\item if $t$ is $(t_1 \backslash t_2)$, then $\val{t}{Z}$ is
		\begin{align*}
			\exists I J Q R \, (I = J \times Q + R &\land \val{t_1}{I} \land \val{t_2}{J}\\
				& \land J \neq 0 \land R \geq 0 \land R < Q \land Z = R),
		\end{align*}
		where $I$, $J$, $Q$, $R$ are fresh integer variables;
	\item if $t$ is $(t_1 \twodots t_2)$, then $\val{t}{Z}$ is
		\[
			\exists I J K \, (\val{t_1}{I} \land \val{t_2}{J} \land I \leq K \land K \leq J \land Z = K),
		\]
		where $I$, $J$, $K$ are fresh integer variables.
\end{itemize}
For example, $\val{\numeral{1} \twodots \numeral{3}}{Z}$ is
\[
	\exists I J K \, (I = \numeral{1} \land J = \numeral{3} \land I \leq K \land K \leq J \land Z = K),
\]
and $\val{X + \numeral{1}}{Z}$ is
\[
	\exists I J \, (Z = I + J \land I = X \land J = \numeral{1}).
\]

The other thing to do in preparation for defining $\tau^*$ is to define the translation~$\tau^B$ that will be applied to expressions in the body of the rule:
\begin{itemize}
	\item $\tau^B(p(t_1, \dots, t_k))$ is
		\[
			\exists Z_1 \dots Z_k \, (\val{t_1}{Z_1} \land \cdots \land \val{t_k}{Z_k} \land p(Z_1, \dots, Z_k));
		\]
	\item $\tau^B(\mynot p(t_1, \dots, t_k))$ is
		\[
			\exists Z_1 \dots Z_k \, (\val{t_1}{Z_1} \land \cdots \land \val{t_k}{Z_k} \land \neg p(Z_1, \dots, Z_k));
		\]
	\item $\tau^B(\mynot \mynot p(t_1, \dots, t_k))$ is
		\[
			\exists Z_1 \dots Z_k \, (\val{t_1}{Z_1} \land \cdots \land \val{t_k}{Z_k} \land \neg \neg p(Z_1, \dots, Z_k));
		\]
	\item $\tau^B(t_1 \prec t_2)$ is
		\[
			\exists Z_1 Z_2 \, (\val{t_1}{Z_1} \land \val{t_2}{Z_2} \land Z_1 \prec Z_2);
		\]
\end{itemize}
where each $Z_i$ is a fresh program variable.

Now, we define
\[
	\tau^*(\head \ar B_1 \land \cdots \land B_n)
\]
as the universal closure of the formula
\[
	\tau^B(B_1) \land \cdots \land \tau^B(B_n) \rar H,
\]
where~$H$ is
\begin{itemize}
	\item $\forall Z_1 \dots Z_k \, (\val{t_1}{Z_1} \land \cdots \land \val{t_k}{Z_k}
		\rar p(Z_1, \dots, Z_k))$

		if $\head$ is $p(t_1, \dots, t_k)$;
	\item $\forall Z_1 \dots Z_k \, (\val{t_1}{Z_1} \land \cdots \land \val{t_k}{Z_k}
		\rar p(Z_1, \dots, Z_k) \lor \neg p(Z_1, \dots, Z_k))$

		if $\head$ is $\{p(t_1, \dots, t_k)\}$;
	\item $\bot$ if $\head$ is empty;
\end{itemize}
where $Z_1, \dots, Z_k$ are fresh program variables.

For any program~$\Pi$, $\tau^* \Pi$ stands for the set of formulas~$\tau^* R$ for all rules~$R$ of~$\Pi$.
For example, the result of applying $\tau^*$ to program~\eqref{ex2} is the pair of sentences
\begin{gather*}
	\top \rar \forall Z \, (\val{\numeral{1} \twodots \numeral{3}}{Z} \rar p(Z) \lor \neg p(Z)),\\
	\forall X \, (\exists Z \, (Z = X \land p(Z)) \rar \forall Z \, (\val{X + \numeral{1}}{Z} \rar q(Z))).
\end{gather*}
 \section{Main Lemma}
\label{sec:main.lemma}

The lemma stated in this section is used in the proofs of Theorems~\ref{th1} and \ref{th2}, which are stated later in Section~\ref{sec:program.completion}.
It relates the stable models of a program~$\Pi$ to the stable models of the first-order theory~$\tau^* \Pi$.

Recall that, for any first-order sentence~$F$ and any list~$\boldp$ of predicate symbols (other than equality), $\SM_\boldp[F]$ stands for the second-order sentence
\[
	F \wedge \neg \exists \boldu((\boldu < \boldp) \wedge F^*(\boldu))
\]
\cite[Section~2.3]{feleli11a}.
The definition of~$\SM_\boldp$ in that paper is given for formulas with one sort of variables, but its extension to many-sorted formulas is straightforward, and we will apply it here to the class of formulas defined in Section~\ref{sec:formulas}.

The models of $\SM_\boldp[F]$ are called the \emph{{$\boldp$-stable} models} of~$F$.
For a set~$\Gamma$ of first-order sentences, the \emph{{$\boldp$-stable} models} of~$\Gamma$ are the $\boldp$-stable models of the conjunction of all formulas in~$\Gamma$.

For any program~$\Pi$ and any set~$\I$ of precomputed atoms, $\Iu$~stands for the interpretation of the signature~$\sigma_\Pi$ such that
\begin{itemize}
	\item the universe of the first sort of~$\Iu$ is the set of all precomputed terms;
	\item the universe of the second sort of~$\Iu$ is the set of all numerals;
\item $\Iu$ interprets every precomputed term~$t$ as~$t$;
	\item $\Iu$ interprets $\numeral{m} + \numeral{n}$ as $\numeral{m + n}$, and similarly for subtraction and multiplication;
	\item $\Iu$ interprets every precomputed atom ${p(t_1, \dotsc, t_n)}$ as true iff it belongs to~$\I$;
	\item $\Iu$ interprets every atomic sentence~${t_1 \prec t_2}$, where~$t_1$ and~$t_2$ are precomputed terms, as true iff the relation~$\prec$ holds for the pair~${(t_1, t_2)}$.
\end{itemize}

\begin{mainlemma}
Let $\Pi$~be a program and let $\boldp$~be the list of all predicate symbols occurring in~$\Pi$ other than the comparison symbols.
A set~$\I$ of precomputed atoms is a stable model of~$\Pi$ iff $\Iu$~is a \SMstable{} model of~$\tau^* \Pi$.
\end{mainlemma}

Proofs are relegated to \ref{appendix:proofs}.
 \section{Programs with Input and Output}
\label{sec:io-programs}

\subsection{Syntax}
\label{ssec:syntax.spec}

A \emph{program with input and output}, or an~\emph{io-program}, is a quadruple
\begin{align}
	\label{iop}
	(\Pi, \PH, \In, \Out),
\end{align}
where
\begin{itemize}
	\item $\Pi$ is a finite set of rules;
	\item $\PH$~is a finite set of symbolic constants, called the \emph{placeholders} of~\eqref{iop};
	\item $\In$~is a finite set of predicate symbols that do not occur in the heads of rules of~$\Pi$; they are called the \emph{input symbols} of~\eqref{iop};
	\item $\Out$~is a finite set of predicate symbols that is disjoint from~$\In$; they are called the \emph{output symbols} of~\eqref{iop}.
\end{itemize}
The input and output symbols of an io-program are collectively called its \emph{public symbols}; the predicate symbols occurring in the rules that are not public are called \emph{private}.
We also distinguish between public and private atoms, depending on whether their predicate symbols are public or private.
Every set~$\I$ of public atoms is the union of two disjoint sets---the set~$\I^{in}$ of atoms in~$\I$ that begin with an input symbol and the set~$\I^{out}$ of atoms that begin with an output symbol.

As an example, consider the io-program $\Omega_1$ with rules~\eqref{ex1}, the placeholder~$n$, the input symbol~$s/2$, and the output symbol~$\incover/1$.
The only private symbol of~$\Omega_1$ is $\covered/1$.

By~$\sigma_\Omega$, we denote the part of the signature $\sigma_\Pi$ (Section~\ref{sec:formulas}) obtained from it by removing all private symbols.
Formulas of that signature will be used in writing specifications.

Programs can be viewed as a special case of io-programs if we agree to identify a program~$\Pi$ with the io-program~\eqref{iop} in which $\PH$ and $\In$ are empty and $\Out$~is the set of all predicate symbols that occur in the rules of~$\Pi$.

\subsection{Semantics}

A~\emph{valuation} for an io-program~$\Omega$ is a function defined on the set of placeholders of~$\Omega$ such that its values are precomputed terms different from placeholders.
An~\emph{input} for~$\Omega$ is a pair $(\boldv, \boldi)$, where~$\boldv$ is a valuation and~$\boldi$ is a set of precomputed atoms that do not contain placeholders, such that the predicate symbol of each atom in~$\boldi$ is an input symbol of~$\Omega$.

For an io-program~$\Omega$ and an input~$(\boldv, \boldi)$, we denote by $\Omega(\boldv, \boldi)$ the program consisting of
\begin{itemize}
	\item the rules obtained from the rules of~$\Omega$ by replacing every occurrence of every placeholder~$t$ by the term~$\boldv(t)$ and
	\item the facts $\boldi$.
\end{itemize}
About a set of precomputed atoms, we say that it is an \emph{io-model} of an io-program~$\Omega$ for an input~$(\boldv, \boldi)$ if it is the set of all public atoms of some stable model of~$\Omega(\boldv, \boldi)$.

For example, we can define an input $(\boldv_1, \boldi_1)$ for~$\Omega_1$ by the conditions
\[
	{\boldv_1}(n) = \numeral{3}, \quad
	\boldi_1 =
		\{
			s(a, \numeral{1}),
			s(b, \numeral{1}),
			s(b, \numeral{2}),
			s(c, \numeral{2}),
			s(c, \numeral{3})
		\}.
\]
With this input, $\Omega_1$~encodes the exact covers of the collection $\{\{a, b\}, \{b, c\}, \{c\}\}$.
The program $\Omega_1(\boldv_1, \boldi_1)$ consists of the rules
\begin{align*}
	&\{\incover(\numeral{1} \twodots \numeral{3})\},\\
	&\ar I \neq J \land \incover(I) \land \incover(J) \land s(X, I) \land s(X, J),\\
	&\covered(X) \ar \incover(I) \land s(X, I),\\
	&\ar s(X, I)\land \mynot \covered(X),\\
	&s(a, \numeral{1}),\ \dots,\ s(c, \numeral{3}).
\end{align*}
It has one stable model,
\[
	\{
		s(a, \numeral{1}), \dots, s(c, \numeral{3}),
		\incover(\numeral{1}), \incover(\numeral{3}),
		\covered(a), \covered(b), \covered(c)
	\},
\]
so that the only io-model of~$\Omega$ is
$\{s(a, \numeral{1}), \dots, s(c, \numeral{3}), \incover(\numeral{1}), \incover(\numeral{3})\}$.

It is clear that atoms in an io-model do not contain private symbols.
Furthermore, if $\I$~is an io-model of~$\Omega$ for an input~$(\boldv, \boldi)$, then $\I^{in} = \boldi$.
This is clear from the fact that the only rules of $\Omega(\boldv, \boldi)$ containing input symbols in the head are the facts~$\boldi$.

IO-models of io-programs also have a characterization in terms of first-order formulas similar to the characterization of stable models in the Main Lemma above (see Appendix~\ref{sec:main-lemma.io-programs}).
The proof of this characterization is based on the aforementioned Main Lemma and results from~\cite{felelipa09a} and~\cite{cafali20a}.
This is later used in the proofs of Theorems~\ref{th1} and~\ref{th2} below.

\subsection{Specifications}
\label{ssec:specifications}

In this paper, we are interested in io-programs written in response to formal specifications.
A specification, as we understand this term, includes a list~$\PH$ of placeholders, a list~$\In$ of input symbols, and a list~$\Out$ of output symbols that a future io-program~$\Omega$ is required to have.
(The choice of private symbols of~$\Omega$ is left to the discretion of the programmer.)
That part of the specification determines the set of valuations and inputs of~$\Omega$ as well as the signature~$\sigma_\Omega$.
A specification also includes two lists of sentences in first-order logic over that signature: a list~$A$ of \emph{assumptions} and a list~$S$ of \emph{specs}.
The sentences in~$A$ do not contain output symbols.
They are intended to describe the class of inputs that the author of the specification is interested in.
The sentences in~$S$ are meant to describe the io-models that the future io-program is required to have for these inputs.

For example, a specification for the exact cover problem (see Footnote~\ref{ft1}) may include the assumptions
\begin{align}
	\label{as.cover}
	\begin{split}
		&\exists N \, (n = N \land N \geq \numeral{0}),\\
		&\forall X Y \, (s(X, Y) \rar \exists I \, (Y = I \land I \geq \numeral{1} \land I \leq n)),
	\end{split}
\end{align}
where $N$ and $I$ are integer variables.
They tell us that a ``good'' input describes the length and composition of a list of sets of precomputed atoms.
Its specs may look like this:
\begin{align}
	\label{eq:example-cover-specs}
	\begin{split}
		&\forall Y \, (\incover(Y) \rar
			\exists I \, (Y = I \land I \geq \numeral{1} \land I \leq n)),\\
		&\forall X \, (\exists Y \, s(X, Y) \rar \exists Y (s(X, Y) \land \incover(Y))),\\
		&\forall Y Z \, (\exists X \, (s(X, Y) \land s(X, Z))
			\land \incover(Y) \land \incover(Z) \rar Y = Z).
	\end{split}
\end{align}
These formulas describe the concept of an exact cover.

To give a precise definition of the semantics of specifications, we need to generalize the process of transforming a set~$\I$ of precomputed atoms into the interpretation~$\Iu$ defined in Section~\ref{sec:main.lemma}.
This generalization takes into account the fact that placeholders may be ``reinterpreted'' by a valuation.
For any set~$\I$ of precomputed public atoms that do not contain placeholders and any valuation~$\boldv$, we denote by $\I^\boldv$ the interpretation of the signature~$\sigma_\Omega$ such that
\begin{itemize}
	\item the universe of the first sort of~$\I^\boldv$ is the set of all precomputed terms;
	\item the universe of the second sort of~$\I^\boldv$ is the set of all numerals;
	\item $\I^\boldv$ interprets every placeholder~$t$ as~$\boldv(t)$ and every other precomputed term~$t$ as~$t$;
	\item $\I^\boldv$ interprets $\numeral{m} + \numeral{n}$ as $\numeral{m + n}$, and similarly for subtraction and multiplication;
	\item $\I^\boldv$ interprets every precomputed atom ${p(t_1, \dotsc, t_n)}$ that does not contain placeholders as true iff it belongs to~$\I$;
	\item $\I^\boldv$ interprets every atomic sentence~${t_1 \prec t_2}$, where~$t_1$ and~$t_2$ are precomputed terms different from placeholders, as true iff the relation~$\prec$ holds for the pair~${(t_1, t_2)}$.
\end{itemize}
In the exact cover example, $\I^\boldv$~interprets the placeholder~$n$ as~$\boldv(n)$.
This interpretation satisfies the first of assumptions~\eqref{as.cover} iff~$\boldv(n)$ is a nonnegative numeral.

We say that~$\Omega$ \emph{implements} the specification with assumptions~$A$ and specs~$S$ if, for every valuation~$\boldv$ and every set~$\I$ of precomputed public atoms without placeholders such that $\I^\boldv$ satisfies~$A$,
\begin{center}
	$\I$ is an io-model of~$\Omega$ for the input $(\boldv, \I^{in})$ iff $\I^\boldv$~satisfies~$S$.
\end{center}

In the next section, we show that this condition can be sometimes established by proving a sentence in a first-order theory.
This is how \anthem{} works.
 \section{Completion and Tightness}
\label{sec:program.completion}

\subsection{Completed Definitions}
\label{ssec:compdef}

The \emph{definition} of a predicate symbol~$p/n$ in an io-program~$\Omega$ is the set of all rules of~$\Omega$ that have the forms
\begin{align}
	\label{def1}
	p(t_1, \dots, t_n) \leftarrow \body
\end{align}
and
\begin{align}
	\label{def2}
	\{p(t_1, \dots, t_n)\} \leftarrow \body.
\end{align}
If the definition of~$p/n$ in~$\Omega$ consists of rules $R_1, \dots, R_k$, then the \emph{formula representations} $F_1, \dots, F_k$ of these rules are constructed as follows.
Take fresh program variables $V_1, \dots, V_n$.
If~$R_i$~is~\eqref{def1}, then $F_i$~is the formula
\[
	\tau^B(\body) \land \val{t_1}{V_1} \land \cdots \land \val{t_n}{V_n}
\]
(see Section~\ref{sec:formulas} for the definitions of $\tau^B$ and $\valop$).
If~$R_i$~is~\eqref{def2}, then $F_i$~is the formula
\begin{align}
	\label{fc}
	\tau^B(\body) \land p(V_1, \dots, V_n) \land \val{t_1}{V_1} \land \cdots \land \val{t_n}{V_n}.
\end{align}

We are ready to define ``completed definitions'' of predicate symbols.
Completed definitions are generally second-order formulas---the private symbols of~$\Omega$ are represented in them by predicate variables.

Let~$\Omega$ be an io-program with private symbols $p_1/n_1, \dots, p_l/n_l$.
Choose pairwise distinct predicate variables $P_1, \dots, P_l$ and take a predicate symbol~$p/n$ that occurs in~$\Omega$ and is different from its input symbols.
The \emph{completed definition} of~$p/n$ in~$\Omega$ is obtained from the formula
\begin{align}
	\label{eq:completed.definition}
	\forall V_1 \ldots V_n \, \left(p(V_1, \dots, V_n) \lrar \bigvee_{i = 1}^k \exists \boldU_i \, F_i\right),
\end{align}
where $\boldU_i$~is the list of all variables occurring in rule~$R_i$ after substituting $P_1, \dots, P_l$ for $p_1/n_1, \dots, p_l/n_l$.

For example, the completed definition of $\incover/1$ in~$\Omega_1$ is
\begin{align}
	\label{def.incover}
	\begin{split}
		\forall V \, (&\incover(V)\\
			&\lrar \top \land \incover(V) \land
				\exists I J K \, (I = \numeral{1} \land J = n \land I \leq K \land K \leq J \land V = K)).
	\end{split}
\end{align}
It can be equivalently rewritten as
\[
	\forall V \, (\incover(V)
		\rar \exists I J K \, (I = \numeral{1} \land J = n \land I \leq K \land K \leq J \land V = K)).
\]

\subsection{Completion}
\label{ssec:completion}

For a constraint $\ar \body$, its \emph{formula representation} is defined as the formula obtained from the universal closure of $\neg \tau^B(\body)$ by substituting $P_1, \dots, P_l$ for $p_1/n_1, \dots, p_l/n_l$.
For example, the formula representation of the second constraint of~$\Omega_1$ is obtained from the formula
\[
	\forall X I \, \neg \tau^B(s(X, I) \land \mynot \covered(X))
\]
by substituting a predicate variable $\Covered$ for the predicate constant $\covered$.
It can be equivalently rewritten as
\begin{align}
	\label{def.constraint}
	\forall X I \, \neg (s(X, I) \land \neg \Covered(X)).
\end{align}

The \emph{completion} $\COMP(\Omega)$ of an io-program~$\Omega$ is the sentence $\exists P_1 \ldots P_l \, F$, where $P_1, \dots, P_l$ are the predicate variables corresponding to the private symbols of~$\Omega$, as described above, and~$F$ is the conjunction of
\begin{itemize}
	\item the completed definitions of all predicate symbols occurring in~$\Omega$ other than its input symbols and
	\item the formula representations of all constraints in~$\Omega$.
\end{itemize}
Note that this sentence does not contain any of the private symbols of~$\Omega$.
In other words, it is a second-order sentence over the signature~$\sigma_\Omega$ (Section~\ref{ssec:syntax.spec}).

For example, the predicate completion of~$\Omega_1$ is
\[
	\exists \Covered \, (F_1 \land F_2 \land F_3 \land F_4),
\]
where~$F_1$ is the completed definition~\eqref{def.incover} of $\incover$, $F_2$ is the completed definition of $\covered$, and $F_3$, $F_4$ are the formula representations of the two constraints.

\begin{theorem}
	\label{th1}
	If~$\I$ is an io-model of an io-program~$\Omega$ for an input~$(\boldv, \boldi)$, then the interpretation~$\I^\boldv$ satisfies $\COMP(\Omega)$.
\end{theorem}

\subsection{Tight Programs}

The \emph{predicate dependency graph} of an io-program $\Omega$ is the directed graph that
\begin{itemize}
	\item has the predicate symbols occurring in rules of~$\Omega$ as its vertices and
	\item has an edge from $p/n$ to~$q/m$ if there is a rule~$R$ in~$\Omega$ such
that~$p/n$ occurs in the head of~$R$ and~$q/m$ occurs in the body of~$R$.
\end{itemize}
The edge from~$p/n$ to~$q/m$ in the predicate dependency graph is \emph{positive} if there is a rule~$R$ in~$\Omega$ such that $p/n$~occurs in the head of~$R$ and~$q/m$ occurs in an atom in the body of~$R$ that is not preceded by $\mynot$.

An io-program is \emph{tight} if its predicate dependency graph has no cycles consisting of positive edges.

\begin{theorem}
	\label{th2}
	For any tight io-program~$\Omega$ and any set~$\I$ of precomputed public atoms without placeholders,~$\I$ is an io-model of~$\Omega$ for an input~$(\boldv, \boldi)$ iff $\I^\boldv$ satisfies $\COMP(\Omega)$ and $\I^{in} = \boldi$.
\end{theorem}

This theorem shows that, in application to tight io-programs, the semantics of specifications (Section~\ref{ssec:specifications}) can be characterized in terms of completion:
A tight io-program~$\Omega$ implements the specification with assumptions~$A$ and specs~$S$ iff, for every valuation~$\boldv$ and every set~$\I$ of precomputed public atoms without placeholders, the interpretation~$\I^\boldv$ satisfies the sentence
\begin{align}
	\label{imp1}
	A \rar (\COMP(\Omega) \lrar S).
\end{align}
(We identify a list of sentences with the conjunction of its members.)

An interpretation of the signature~$\sigma_\Omega$ is \emph{standard} if it can be represented in the form~$\I^\boldv$ for some set~$\I$ of precomputed public atoms without placeholders and some valuation~$\boldv$.
We can then say that a tight io-program~$\Omega$ implements the specification with assumptions~$A$ and specs~$S$ iff sentence~\eqref{imp1} is satisfied by all standard interpretations.

The automation of reasoning about formulas of form~\eqref{imp1} is difficult because $\COMP(\Omega)$ has, generally, bound predicate
variables.
We will now discuss a special case when reasoning about such formulas can be reduced to first-order reasoning.

\subsection{Private Recursion}
\label{ssec:private-recursion}

Recall that the completion of an io-program~$\Omega$ is a sentence of the form $\exists \boldP \, F$, where $\boldP$ is a list $P_1, \dots, P_l$ of predicate variables (Section~\ref{ssec:completion}).
That list has the same length~$l$ as the list~$\boldp$ of the private symbols $p_1/n_1, \dots, p_l/n_l$ of~$\Omega$.
Formula~$F$ is a conjunction that includes, among its conjunctive terms, the completed definitions $F_1(\boldP), \dots, F_l(\boldP)$ of $p_1/n_1, \dots, p_l/n_l$.
Denote the conjunction of all other conjunctive terms of~$F$ by~$F'(\boldP)$.
Then, the completion of~$\Omega$ can be written as
\begin{gather}
	\exists \boldP \, (F_1(\boldP) \land \cdots \land F_l(\boldP) \land F'(\boldP)).
	\label{eq:private.recursion.existencial}
\end{gather}

About an io-program~$\Omega$, we say that it \emph{uses private recursion} if
\begin{itemize}
	\item its predicate dependency graph has a cycle such that every vertex in it is a private symbol or
	\item it includes a choice rule with a private symbol in the head.
\end{itemize}

\begin{theorem}
	\label{th3}
	If an io-program~$\Omega$ does not use private recursion, then its completion is equivalent to
	\begin{gather}
		\forall \boldP \, (F_1(\boldP) \land \cdots \land F_l(\boldP) \rar F'(\boldP)).
		\label{eq:private.recursion.universal}
	\end{gather}
\end{theorem}

From this theorem, we can conclude that, if~$\Omega$ does not use private recursion, sentence~\eqref{imp1} is equivalent to
\begin{align*}
	&(A \rar (\exists \boldP \, (F_1(\boldP) \land \cdots \land F_l(\boldP) \land F'(\boldP)) \rar S))\\
	\land \, &(A \rar (S \rar \forall \boldP \, (F_1(\boldP) \land \cdots \land F_l(\boldP) \rar F'(\boldP))))
\end{align*}
and consequently to
\begin{align}
	\label{imp2}
	\forall \boldP \, (A \land F_1(\boldP) \land \cdots \land F_l(\boldP) \rar (F'(\boldP) \lrar S)).
\end{align}

The procedure used by \anthem{} for program verification is applicable to an io-program if it is tight and does not use private recursion.
The procedure is based on the following idea.
We choose a set of \emph{axioms}---sentences over the signature~$\Omega$ that are satisfied by all standard interpretations.
This set may include, for instance, the formulas $\forall N \, (N + \numeral{0} = N)$ and $\numeral{2} \times \numeral{2} = \numeral{4}$.
It may include also the formula $a \neq b$, where $a$~and $b$ are distinct symbolic constants, unless at least one of them is a placeholder.
To establish that $\Omega$~implements the given specification, \anthem{} verifies that the axioms entail sentence~\eqref{imp2}.
This can be accomplished by first-order reasoning---by verifying that the axioms entail the sentence obtained from~\eqref{imp2} by replacing the bound predicate variables with distinct predicate constants that do not occur in the axioms.
The predicate constants used by \anthem{} are actually the private predicate symbols from the rules of~$\Omega$, so that the formula derived by \anthem{} from the axioms is
\[
	A \land F_1(\boldp) \land \cdots \land F_l(\boldp) \rar (F'(\boldp) \lrar S).
\]
Proof search is performed by the resolution theorem prover \vampire.
The task is divided into two parts: deriving the specs~$S$ from the axioms, the assumptions $A$, and the formulas $F_1(\boldp), \dots, F_l(\boldp), F'$ (``forward direction'') and deriving~$F'$ from the axioms, the assumptions $A$, and the formulas $F_1(\boldp), \dots, F_l(\boldp), S$ (``backward direction'').
 \section{Verifying Tight Logic Programs with \anthem{}}
\label{sec:anthem}

Let us return to the exact cover problem, which we encoded in rules~\eqref{ft1} and represented as io-program~$\Omega_1$ earlier.
We would like to show that this program implements specification~\eqref{eq:example-cover-specs} under the assumptions stated in~\eqref{as.cover}.
This verification can be automated with \anthem{}~1.0.\footnote{\url{https://github.com/potassco/anthem/releases}}
To this effect, \anthem{} reads a logic program and a specification, checks that the program is tight and does not contain private recursion, generates the program’s completion, and then performs automated proof searches trying to establish that the program actually implements the given specification.

As the input program, we can simply pass rules~\eqref{ft1} to \anthem{} in the input language of \gringo.
For the specification, \anthem{} accepts files in a custom \emph{specification file} format.
Specification files may contain multiple types of statements.
\texttt{input} and \texttt{output} statements specify input predicates and placeholders as opposed to output predicates.
Axioms, assumptions and specs are expressed by \texttt{axiom}, \texttt{assume}, and \texttt{spec} statements, respectively.
\anthem’s specification file format has a simple syntax for such formulas.
This format is a concise form of the segment of the TPTP language \cite{sutcliffe17a} that is required for expressing specifications.

Similar to the input language of \gringo, identifiers starting with a lowercase letter are interpreted as predicate and object constants, while variables are denoted by identifiers starting with an uppercase letter.
An important distinction is that variables starting with letters \texttt{I}, \texttt{J}, \texttt{K}, \texttt{L}, \texttt{M}, or \texttt{N} are implicitly assumed to be integer variables---as opposed to program variables, which must start with \texttt{U}, \texttt{V}, \texttt{W}, \texttt{X}, \texttt{Y}, or \texttt{Z}.
\anthem{} shows an error message if variables are used that start with letters not listed above.
Terms may contain integer arithmetics with a syntax close to \gringo.
In addition to that, the keywords \texttt{not}, \texttt{and}, \texttt{or}, \texttt{->}, \texttt{<->}, \texttt{exists}, and \texttt{forall} denote the common operators of first-order logic.

We can encode specification~\eqref{eq:example-cover-specs} and assumptions~\eqref{as.cover} in this format, as shown in Figure~\ref{fig:exact-cover-specification}.
\begin{figure}[tp]
	\begin{lstlisting}[language = anthem]
input: n -> integer, s/2.
output: in_cover/1.
assume: n >= 0.
assume: forall Y (exists X s(X, Y) -> exists I (Y = I and I >= 1 and I <= n)).
spec: forall Y (in_cover(Y) -> exists I (Y = I and I >= 1 and I <= n)).
spec: forall X (exists Y s(X, Y) -> exists Y (s(X, Y) and in_cover(Y))).
spec: forall Y, Z (exists X (s(X, Y) and s(X, Z)) and in_cover(Y) and in_cover(Z)
                   -> Y = Z).
	\end{lstlisting}
	\caption
	{
		A specification file to verify that a program encoding the exact cover problem fulfills specification~\eqref{eq:example-cover-specs} under assumptions~\eqref{as.cover}.
		First, $n$ is declared a placeholder using an \texttt{input} statement.
		The \texttt{-> integer} notation instructs \anthem{} to represent $n$ by an integer variable internally.
		Next, the input and output predicates of $\Omega_1$ ($s/2$ and $in/1$, respectively) are declared.
		The \texttt{assume} and \texttt{spec} statements encode assumptions~\eqref{as.cover} and specs~\eqref{eq:example-cover-specs}.
	}
	\label{fig:exact-cover-specification}
\end{figure}
If we give this specification to \anthem{} along with the exact cover program encoded in~\eqref{ft1}, \anthem{} first generates the program’s completion and shows it to the user (see Figure~\ref{fig:exact-cover-completion-by-anthem}).
\begin{figure}[tp]
	\begin{lstlisting}[language = anthem]
forall X (covered(X) <-> exists N (in_cover(N) and s(X, N)))
forall X (in_cover(X) <-> in_cover(X) and exists N (1 <= N and N <= n and X = N))
forall N1, N2, X not (N1 != N2 and in_cover(N1) and in_cover(N2)
                      and s(X, N1) and s(X, N2))
forall X, N not (s(X, N) and not covered(X))
	\end{lstlisting}
	\caption
	{
		\anthem{} generates the completion of the input program---here, $\Omega_1$, the exact cover problem encoded by rules~\eqref{ft1}.
		The first two lines contain the completed definitions of $\covered/1$ and $\incover/1$.
		They are followed by the two translated constraints.
		Note that \anthem{} applies basic equivalent transformations to simplify the completion formulas.
		For example, the completed definition of $\incover$ shown here corresponds to formula~\eqref{def.incover}, but the integer variables $I$ and $J$ were eliminated by \anthem{}.
		The last line matches formula~\eqref{def.constraint}, the formula representation of the second constraint of $\Omega_1$.
	}
	\label{fig:exact-cover-completion-by-anthem}
\end{figure}
This is done as described in Section~\ref{ssec:completion}.
Afterward, \anthem{} verifies very quickly that the program indeed implements the specification in Figure~\ref{fig:exact-cover-specification} using the theorem prover \vampire{} and informs the user accordingly.\footnote
{
	When running \vampire~4.4 with the options \texttt{-{-}mode casc} and \texttt{-{-}cores 8} on a computer with an Intel Core i7-7700K CPU and 16~GB RAM, the equivalence proof completes within a total of 0.37 seconds.
}

To accomplish this, \anthem{} splits the equivalence verification task into two parts, as mentioned at the end of Section~\ref{ssec:private-recursion}.
In either direction, \anthem{} starts with a set of presupposed formulas and performs one proof step for each formula to verify.
For each proof step, \anthem{} devises a program in the TFF (“typed first-order formula”) segment of the TPTP language.
In this TPTP program, all presupposed formulas are represented as TPTP axioms, and the formula to verify in this proof step is used as the conjecture.
\anthem{} then passes this TPTP program to \vampire{}.
If a proof step is successful within a predefined time limit, the formula is added to the list of presupposed formulas and the next proof step is conducted.
The fact that \anthem{} uses the TPTP format to communicate with \vampire{} is completely hidden from the user.

In some cases, \vampire{} is incapable of verifying the equivalence of a specification to a program directly and within a reasonable time frame, for which there are two possible explanations.
First, the set of axioms known to \vampire{} might not be sufficient to perform the proof, and second, it might be necessary to further guide \vampire’s proof search by providing helper lemmas.
To that end, \anthem{} allows users to supply additional axioms and helper lemmas with the \texttt{axiom} and \texttt{lemma} directives, respectively.
\anthem{} verifies helper lemmas in the given order before the specs or the completion of the program.
The usage of axioms and helper lemmas can be seen in a second example, a specification for programs that compute the floor of the square root of an integer~$n$ shown in Figure~\ref{fig:example-square-root-specification} and the implementation in Figure~\ref{fig:example-square-root-program}.
\begin{figure}[tp]
	\begin{lstlisting}[language = anthem]
input: n -> integer.
assume: n >= 0.
output: q/1.
spec: exists N (forall X (q(X) <-> X = N) and N >= 0
                and N * N <= n and (N + 1) * (N + 1) > n).
	\end{lstlisting}
	\caption
	{
		A specification for programs calculating the floor of the square root of an integer~$n$.
	}
	\label{fig:example-square-root-specification}
\end{figure}
\begin{figure}[tp]
	\begin{lstlisting}[language = clingo]
p(X) :- X = 0..n, X * X <= n.
q(X) :- p(X), not p(X + 1).
	\end{lstlisting}
	\caption
	{
		A program meant to implement the specification in Figure~\ref{fig:example-square-root-specification}.
	}
	\label{fig:example-square-root-program}
\end{figure}
\anthem{} and \vampire{} verify the equivalence between this specification and program only when supplying an axiom unknown to \vampire{} and a list of helper lemmas (shown in Figure~\ref{fig:example-square-root-lemmas}).
\begin{figure}[tp]
	\begin{lstlisting}[language = anthem]
axiom: p(0) and forall N (N >= 0 and p(N) -> p(N + 1))
       -> forall N (N >= 0 -> p(N)).

lemma(forward): forall X (p(X) <-> exists N (X = N and N >= 0 and N * N <= n)).
lemma(forward): forall X (q(X) <-> exists N2 (X = N2 and N2 >= 0 and N2 * N2 <= n
                                              and (N2 + 1) * (N2 + 1) > n)).
lemma(forward): forall N1, N2 (N1 >= 0 and N2 >= 0 and N1 < N2
                               -> N1 * N1 < N2 * N2).
lemma(forward): forall N (N >= 0 and p(N + 1) -> p(N)).
lemma(forward): not p(n + 1).
lemma(forward): forall N1, N2 (q(N1) and N2 > N1 -> not q(N2)).
lemma(forward): forall N (N >= 0 and not p(N + 1) -> (N + 1) * (N + 1) > n).

lemma(backward): forall N1, N2 (N1 >= 0 and N2 >= 0 and N1 * N1 <= N2 * N2
                                -> N1 <= N2).
lemma(backward): forall N (p(N) <-> 0 <= N and N <= n and N * N <= n).
lemma(backward): forall N (not p(N) and N >= 0 -> N * N > n).
lemma(backward): forall N (N >= 0 -> N * N < (N + 1) * (N + 1)).
lemma(backward): forall N1, N2 (p(N1) and not p(N2) and N2 >= 0 -> N1 < N2).
lemma(backward): forall N1, N2 (p(N1) and not p(N1 + 1) and p(N2)
                                and not p(N2 + 1) -> N1 = N2).
	\end{lstlisting}
	\caption
	{
		\vampire{} is unable to verify the equivalence of the specification in Figure~\ref{fig:example-square-root-specification} to the program in Figure~\ref{fig:example-square-root-program} in a reasonable time frame.
		This can be addressed by supplying helper lemmas using the \texttt{lemma} keyword, where \texttt{forward} and \texttt{backward} instruct \anthem{} only to apply lemmas in the respective part of the equivalence proof.
		With these helper lemmas, \vampire{} can perform the backward direction of the equivalence proof task but is still unable to verify the forward direction.
		To address this, the set of axioms is manually extended by a fact that is not known to \vampire{}, an instance of induction.
		Such axioms can be added with \texttt{axiom} statements.
	}
	\label{fig:example-square-root-lemmas}
\end{figure}

While \anthem~1.0 introduces the mode of operation described in the present paper, earlier versions of \anthem{} allowed users to try to prove that the stable models of a (not necessarily tight) program satisfy a specification~\cite{lilusc20a}.
This functionality is retained as a second mode of operation in \anthem~1.0.
 \section{Conclusions}
\label{sec:conclusions}

This paper has four main contributions.
First, we presented an extended definition of completion for programs with input and output in a subset of the input language of \gringo{}.
Such programs are similar to lp-functions, defined by Gelfond and Przymusinska~\citeyear{gelprz96a}, except that we allow an input to include the values of placeholders.
Also, they are closely related to \smodels{} program modules \cite{oikjan09a}.
Second, we stated a result describing the relationship between stable models and completion in this context.
It uses the theory of stable models of infinitary formulas by \cite{truszczynski12a} and \cite{halipeva17a}.
Third, we proposed a new theory of formal specifications for programs in the input language of \gringo{}.
Fourth, we presented preliminary experiments with the new version of \anthem{}, a software tool used in combination with a theorem prover to verify that programs with input and output implement a specification.
Only few theorem provers support both the typed first-order form of TPTP and integer arithmetics, and \vampire{} was chosen because its performance outclassed other provers in our early experiments.
In later experiments, \textsc{cvc4} emerged as a possible alternative with similar performance to \vampire, and we consider using \textsc{cvc4} in the future.
As another direction of future work, we would like to extend the subset of the input language of \gringo{} supported by \anthem{}, in particular, by simple forms of aggregates.
 \section*{Acknowledgments}

We thank Michael Gelfond and the anonymous referees for their valuable comments.
Jorge Fandinno is funded by the Alexander von Humboldt Foundation.
 
\bibliographystyle{include/latex-class-tlp/acmtrans}

\clearpage
\appendix

\section{Proofs}
\label{appendix:proofs}

\subsection{Main Lemma}

\label{sec:proof:thm:firt.order.representation.correspondence}

Let us now extend the correspondence between stable models defined in terms of the SM operator and infinitary logic \cite{truszczynski12a} to the above two-sorted case.
We also allow formulas to contain extensional predicate symbols, which are not considered in \cite{truszczynski12a}.

We use the following notation.
By~$\Num$, we denote the set of all numerals, and by~$\T$, we denote the set of all precomputed terms.
For an interpretation~$I$ and a list~$\boldp$ of predicate symbols,
by~$\At{I}$, we denote the set of precomputed atoms $p(t_1,\dotsc,t_k)$ satisfied by~$I$
where ${p \in \boldp}$.

Let~$\boldp,\boldq$ be a partition of the predicate symbols in the signature.
Then,
the \emph{grounding of a sentence~$F$ with respect to an interpretation~$I$ and a set of intensional predicate symbols~$\boldp$} (and extensional predicate symbols~$\boldq$) is defined as follows:
\begin{itemize}
\item $\gr{I}{\bot} = \bot$;
\item for~$p \in \boldp$, $\gr{I}{p(t_1,\dotsc,t_k)} = p((t_1^I)^*,\dotsc,(t_k^I)^*)$;
\item for~$p \in \boldq$, $\gr{I}{p(t_1,\dotsc,t_k)} = \top$ if ${ p((t_1^I)^*,\dotsc,(t_k^I)^*) \in I^\boldq}$

and ${\gr{I}{p(t_1,\dotsc,t_k)} = \bot}$ otherwise;

\item $\gr{I}{t_1 = t_2} = \top$ if $t_1^I = t_2^I$ and $\bot$ otherwise;
\item $\gr{I}{F \otimes G} = \gr{I}{F} \otimes \gr{I}{G}$ if $\otimes$ is $\wedge$, $\vee$, or $\to$;
\item $\gr{I}{\exists X \, F(X)} = \{ \gr{I}{F(u)} \mid u \in \T \}^{\vee}$ if $X$ is a program variable;
\item $\gr{I}{\forall X \, F(X)} = \{ \gr{I}{F(u)} \mid u \in \T \}^{\wedge}$ if $X$ is a program variable;
\item $\gr{I}{\exists X \, F(X)} = \{ \gr{I}{F(u)} \mid u \in \Num \}^{\vee}$ if $X$ is an integer variable;
\item $\gr{I}{\forall X \, F(X)} = \{ \gr{I}{F(u)} \mid u \in \Num \}^{\wedge}$ if $X$ is an integer variable.
\end{itemize}
For a theory~$\Gamma$,
we define~$\gr{I}{\Gamma} = \{ \gr{I}{F} \mid F \in \Gamma \}$.

\begin{definition}
Let $\Gamma$~be a theory and $\boldp$~be a list of predicate symbols.
Then,
an interpretation~$I$ is called an \emph{\INFstable{} model} of~$\Gamma$ if
$\At{I}$~is a stable model of~$\gr{I}{\Gamma}$ in the sense of Definition~1 in \cite{truszczynski12a}.
\end{definition}

Any term, formula, or theory over the two-sorted signature~$\sigma_\Pi$ can be seen as one-sorted if we do not assign sorts to variables. On the other hand, some one-sorted terms and formulas cannot be viewed as terms or formulas over~$\sigma_\Pi$; for instance, the one-sorted term $X + Y$, where $X$ and~$Y$ are program variables, is not a term over~$\sigma_\Pi$.  We will refer to terms, formulas, and theories over~$\sigma_\Pi$ as two-sorted.
\emph{One-sorted interpretations} are defined as usual in first-order logic,
with integer and program variables ranging over the same domain.
\emph{One-sorted \SMstable{} models} and \emph{one-sorted \INFstable{} models}
are defined as \SMstable{} models and \INFstable{} models (see Section~\ref{sec:main.lemma}), respectively, but using one-sorted interpretations rather than two-sorted ones.
We also say that two theories are \emph{one-sorted-equivalent} if both theories have exactly the same one-sorted models.
The following is a special case of Theorem~5 in \cite{truszczynski12a}
restricted to our one-sorted language.

\begin{proposition}\label{prop:SM-INF.one-sorted}
Let $\Gamma$~be a finite theory and $\boldp$~be
the list of all predicate symbols in some signature~$\sigma_\Pi$.
Then, the one-sorted \SMstable{} models of~$\Gamma$ and its one-sorted \INFstable{} models coincide.
\end{proposition}

In the following, we extend this result to two-sorted stable models with extensional predicate symbols (see Proposition~\ref{prop:SM-INF} below).
This requires the following notation and auxiliary results.
The expression $\isint(t)$, where $t$~is a one-sorted term, stands for the one-sorted formula ${t+\overline{1} \neq t+\overline{2}}$.
Given a two-sorted sentence~$F$,
we write~$\unsorted{F}$ to denote the one-sorted sentence resulting from restricting all quantifiers that bind integer variables in~$F$ to~$\isint(t)$.
Formally,
formula~$\unsorted{F}$
is recursively defined as follows:
\begin{itemize}
\item $\unsorted{F} = F$ for any atomic formula~$F$;

\item $\unsorted{(F \otimes G)} = \unsorted{F} \otimes \unsorted{G}$ with~$\otimes \in \{\wedge,\vee,\to\}$;

\item $\unsorted{(\forall X \, F(X))} = \forall X \, \unsorted{F(X)}$;

\item $\unsorted{(\exists X \, F(X))} = \exists X \, \unsorted{F(X)}$;

\item $\unsorted{(\forall N \, F(N))} = \forall N \, (\isint(N) \to \unsorted{F(N)})$;

\item $\unsorted{(\exists N \, F(N))} = \exists N \, (\isint(N) \wedge \unsorted{F(N)})$;
\end{itemize}
where $X$~and $N$~are program variables and integer variables, respectively.
We also define~$\unsorted{\Gamma} = \{ \unsorted{F} \mid F \in \Gamma \}$.

Intuitively, the one-sorted models of~$\unsorted{\Gamma}$
are in a one-to-one correspondence with the two-sorted models of~$\Gamma$.
This one-to-one correspondence is formalized as follows.
The \emph{generalized value} of a ground term is its value if it exists and a fixed (arbitrarily chosen) symbolic constant~$u$ otherwise.
Given a two-sorted interpretation~$I$,
by~$\unsorted{I}$, we denote the one-sorted interpretation such that
\begin{itemize}
\item the universe of~$\unsorted{I}$ is the set of all precomputed terms;
\item $\unsorted{I}$~interprets each ground term as its generalized value;
\item $\unsorted{I}$~interprets every predicate symbol in the same way as~$I$.
\end{itemize}

\begin{lemma}\label{lem:restricting.quantifiers.fo}
Let~$F$ be a two-sorted sentence
and~$I$ be a two-sorted interpretation.
Then, $I \models F$ iff~$\unsorted{I} \models \unsorted{F}$.
\end{lemma}

\begin{proof}
By structural induction.
In case that $F$ is an atomic formula of the form~$p(t_1,\dotsc,t_n)$,
it follows that
$\unsorted{I} \models \unsorted{F}$
iff $(s_1^*,\dotsc,s_n^*) \in p^I$,
where each~$s_i$ is the generalized value of~$t_i$,
iff
$((t_1^I)^*,\dotsc,(t_n^I)^*) \in p^I$
iff
${I \models F}$.
Note that, since~$F$ is a two-sorted sentence,
the generalized value of~$t_i$ coincides with its value.

The only remaining relevant cases are quantifiers over integer variables.
We show here the case of a universal quantifier.
Let~$N$ be an integer variable.
Then,
\begin{align*}
&{\unsorted{I} \models \unsorted{(\forall N\, G(N))}}\\
\textrm{iff } &{\unsorted{I} \models \forall N\, (\isint(N) \to \unsorted{G(N)})}\\
\textrm{iff } &{\unsorted{I} \models \isint(u) \to \unsorted{G(u)}} \textrm{ for all } u \in \T\\
\textrm{iff } &{\unsorted{I} \models \unsorted{G(u)}} \textrm{ for all } u \in \Num\\
\textrm{iff } &{I\models G(u)} \textrm{ for all } u \in \Num \tag*{(induction hypothesis)}\\
\textrm{iff } &{I\models \forall N\, G(N)}.
\end{align*}
The case for the existential quantifier is analogous.
\end{proof}

We extend this result also to the stable models of a first-order formula.
The following auxiliary result is useful for that purpose.

\begin{lemma}[Lemma~5 in \citeNP{feleli11a}]\label{lem:star.implication}
The formula $(\boldu \leq \boldp) \wedge (F^*(\boldu) \to F)$
is satisfied by all one-sorted interpretations
and for any one-sorted formula~$F$.
\end{lemma}

\begin{lemma}\label{lem:restricting.quantifiers.fo.sm}
Let $F$~be a two-sorted sentence,
$I$~be a two-sorted interpretation,
and $\boldp$~be a list of predicate symbols.
Then, $I \models \SM_\boldp[F]$ iff $\unsorted{I} \models \SM_\boldp[\unsorted{F}]$.
\end{lemma}

\begin{proof*}
From Lemma~\ref{lem:restricting.quantifiers.fo},
we get that
$I \models \SM[F]$ iff $\unsorted{I} \models \unsorted{(\SM[F])}$.
We show below that formula ${(\boldu \leq \boldp) \wedge \unsorted{(F^*(\boldu))}}$
is equivalent to ${(\boldu \leq \boldp) \wedge(\unsorted{F}(\boldu))^*}$.
This immediately implies that $\unsorted{(\SM[F])}$
and $\SM[\unsorted{F}]$ are also equivalent.
The proof follows by structural induction,
and the only relevant cases are, again, quantifiers over integer variables.
We show the case of a universal quantifier here.
Let $N$~be an integer variable
and assume ${(\boldu \leq \boldp)}$.
Then,
\begin{align*}
&(\unsorted{(\forall N \, G(N,\boldu))})^*\\
={} &(\forall N \, (\isint(N) \to \unsorted{G(N,\boldu)}))^*\\
={} &\forall N \, (\isint(N) \to \unsorted{G(N,\boldu)})^*\\
={} &\forall N \, ((\isint(N) \to \unsorted{G(N)}) \wedge (\isint(N)^* \to (\unsorted{G(N,\boldu)})^*))\\
={} &\forall N \, ((\isint(N) \to \unsorted{G(N)}) \wedge (\isint(N) \to (\unsorted{G(N,\boldu)})^*))\\
\Leftrightarrow{}
&\forall N \, (\isint(N) \to (\unsorted{G(N)} \wedge (\unsorted{G(N,\boldu)})^*))\\
\Leftrightarrow{}
&\forall N \, (\isint(N) \to (\unsorted{G(N,\boldu)})^*)
\tag{Lemma~\ref{lem:star.implication}}\\
\Leftrightarrow{}
&\forall N \, (\isint(N) \to \unsorted{(G(N,\boldu)^*)})
\tag{induction hypothesis}\\
={} &\unsorted{(\forall N \, G (N,\boldu)^*)}.
\tag*{\proofbox}
\end{align*}
\end{proof*}

This correspondence can also be established in terms of groundings as follows.
The expression~$\Gamma \equivs \Gamma'$,
where $\Gamma$ and~$\Gamma'$ are two infinitary propositional theories, stands for \emph{strong equivalence} in the sense of \cite[Section~3.1]{halipeva17a}.

\begin{lemma}\label{lem:restricting.quantifiers.inf}
Let~$\Gamma$ be a finite two-sorted theory, $I$ be a two-sorted interpretation,
and~$\boldp$ be a list of predicate symbols.
Then,
$\gr{I}{\Gamma} \equivs \gr{\unsorted{I}}{\unsorted{\Gamma}}$.
\end{lemma}
\begin{proof}
We show $\gr{I}{F} \equivs \gr{\unsorted{I}}{\unsorted{F}}$ for a formula~$F$,
which implies $\gr{I}{\Gamma} \equivs \gr{\unsorted{I}}{\unsorted{\Gamma}}$.
We proceed by structural induction.
The only relevant cases are quantifiers over integer variables.
We show the case of a universal quantifier here.
Let~$N$ be an integer variable
and~$G'(u)$ stand for~$\gr{\unsorted{I}}{\unsorted{G}(u)}$.
Then,
\begin{align*}
\gr{I}{\forall N \, G(N)}
&=
\{ \gr{I}{G(u)} \mid u \in \Num \}^{\wedge}
\\
&\equivs \{ G'(u) \mid u \in \Num \}^{\wedge}	\tag{induction hypothesis}
\\
&\equivs (\{ \top \to G'(u) \mid u \in \Num \} \cup \{ \bot \to G'(u) \mid u \in \T \setminus \Num \})^{\wedge}
\\
&=\gr{\unsorted{I}}{\forall N \, (\isint(N) \to \unsorted{G}(N))}
\\
&=\gr{\unsorted{I}}{\unsorted{(\forall N \, G(N))}}.
\end{align*}
The case for the existential quantifier is analogous.
\end{proof}

Next, we combine these results to establish a correspondence between the stable models of a first-order theory and the stable models of its infinitary grounding.

\begin{lemma}\label{lem:aux:thm:SM-INF}
Let~$\Gamma$ be a finite two-sorted theory,
$I$~be a two-sorted interpretation,
and~$\boldp$ be the list of all predicate symbols in the signature.
Then, the \SMstable{} and the \INFstable{} models of~$\Gamma$ coincide.
\end{lemma}

\begin{proof}
We have
\begin{align*}
& I \textrm{ is a $\SMstable{}$ model of } \Gamma\\
\textrm{iff }
&\unsorted{I} \textrm{ is a $\SMstable{}$ model of } \unsorted{\Gamma} \tag*{(Lemma~\ref{lem:restricting.quantifiers.fo.sm})}\\
\textrm{iff }
&\unsorted{I} \textrm{ is an \INFstable{} model of } \unsorted{\Gamma} \tag*{(Proposition~\ref{prop:SM-INF.one-sorted})}\\
\textrm{iff }
&\At{(\unsorted{I})} \textrm{ is a stable\ model of } \gr{\unsorted{I}}{\unsorted{\Gamma}}\\
\textrm{iff }
&\At{(\unsorted{I})} \textrm{ is a stable\ model of } \gr{I}{\Gamma} \tag*{(Lemma~\ref{lem:restricting.quantifiers.inf})}\\
\textrm{iff }
&\At{I} \textrm{ is a stable\ model of } \gr{I}{\Gamma}\\
\textrm{iff }
&I \textrm{ is an \INFstable{} model of } \Gamma.
\end{align*}
For the next-to-last equivalence,
just note that ${\At{(\unsorted{I})} = \At{I}}$.
\end{proof}

\begin{proposition}\label{prop:SM-INF}
For any finite two-sorted theory~$\Gamma$ and list of predicate symbols~$\boldp$, its \SMstable{} models and its \INFstable{} models coincide.
\end{proposition}

\begin{proof}
Let $\boldq$~be the list of all extensional predicate symbols in~$\Gamma$,
that is, all predicate symbols in the signature that do not belong to~$\boldp$,
and let $\Choice{\boldq}$ be the set containing a choice sentence ${\forall\boldU \, (p(\boldU) \vee \neg p(\boldU))}$ for every predicate ${p \in \boldq}$
and where $\boldU$~is a list of distinct program variables.
Let $\Gamma'$~be the theory obtained by replacing each occurrence of~$p(\boldt)$ in~$\Gamma$ with~$p \in \boldq$ by~$\neg\neg p(\boldt)$.
Let $\Gamma_1 = \Gamma \cup \Choice{\boldq}$ and $\Gamma_1' = \Gamma' \cup \Choice{\boldq}$.
Given the choice sentences in~$\Gamma_1$ and~$\Gamma_1'$ for the predicate symbols in~$\boldq$,
the $\boldp\boldq$-stable models of~$\Gamma_1$ and~$\Gamma_1'$ coincide.
Then,
\begin{align*}
&I \textrm{ is a $\boldp$-stable model of } \Gamma\\
\textrm{iff }
&I \textrm{ is a $\boldp\boldq$-stable model of } \Gamma_1 \tag*{(Theorem~2 in \citeNP{feleli11a})}\\
\textrm{iff }
&I \textrm{ is a $\boldp\boldq$-stable model of } \Gamma_1'\\
\textrm{iff }
&I \textrm{ is an INF-$\boldp\boldq$-stable model of } \Gamma_1' \tag*{(Lemma~\ref{lem:aux:thm:SM-INF})}\\
\textrm{iff }
&I \textrm{ is an \INFstable{} model of } \Gamma. \tag*{(see below)}
\end{align*}
It remains to be shown that the
INF-$\boldp\boldq$-stable models of~$\Gamma_1'$
coincide with the
\INFstable{} models of~$\Gamma$.
For this, note that
\begin{align*}
[\grp{\boldp\boldq}{I}{\Gamma_1'}]^{I^{\boldp\boldq}}
&=
[\grp{\boldp\boldq}{I}{\Gamma' \cup \Choice{\boldq}}]^{I^{\boldp\boldq}}
\\
&= [\grp{\boldp\boldq}{I}{\Gamma'}]^{I^{\boldp\boldq}} \cup [\grp{\boldp\boldq}{I}{\Choice{\boldq}}]^{I^{\boldp\boldq}}
\\
&\equiv [\gr{I}{\Gamma'}]^{I^{\boldp\boldq}} \cup [\grp{\boldp\boldq}{I}{\Choice{\boldq}}]^{I^{\boldp\boldq}}
\\
&= [\gr{I}{\Gamma'}]^{I^{\boldp}} \cup [\grp{\boldp\boldq}{I}{\Choice{\boldq}}]^{I^{\boldq}}
\\
&\equiv [\gr{I}{\Gamma'}]^{I^{\boldp}} \cup I^{\boldq}.
\end{align*}
The first two equalities hold by definition.
The third step holds because all predicate symbols in~$\boldq$ occur in~$\Gamma'$ under the scope of negation.
Note that, for $q \in \boldq$,
it follows that
\begin{align*}
[\grp{\boldp\boldq}{I}{\neg q(\boldt)}]^{I^{\boldp\boldq}}
&=
[\neg q((\boldt^I)^*)]^{I^{\boldp\boldq}}
= I^{\boldp\boldq}(\neg q((\boldt^I)^*))
\\
&\equiv \neg I(q((\boldt^I)^*))
\\
&= [\neg I(q((\boldt^I)^*)) ]^{I^{\boldp\boldq}}
= [\neg \grp{\boldp}{I}{ q(\boldt)}]^{I^{\boldp\boldq}}
= [\grp{\boldp}{I}{\neg q(\boldt)}]^{I^{\boldp\boldq}},
\end{align*}
where
\begin{align*}
I(q((\boldt^I)^*)) &= \begin{cases}
\top &\text{if $q((\boldt^I)^*) \in I^\boldq$;}
\\
\bot &\text{otherwise}
\end{cases}
\\
I^{\boldp\boldq}(\neg q((\boldt^I)^*)) &= \begin{cases}
\phantom{\neg}\bot \equiv \neg I(q((\boldt^I)^*))  &\text{if $q((\boldt^I)^*) \in I^\boldq$;}
\\
\neg \bot = \neg I(q((\boldt^I)^*)) &\text{otherwise.}
\end{cases}
\end{align*}
The fourth case is because
no predicate in~$\boldq$ occurs in~$\gr{I}{\Gamma'}$.
Recall that extensional predicate symbols are removed by grounding.
Similarly,
${\Choice{\boldq}}$
only contains predicate symbols from~$\boldq$.

We now prove the following equivalence:
\begin{gather}
[\grp{\boldp\boldq}{I}{\Gamma_1'}]^{I^{\boldp\boldq}}
\equiv
[\gr{I}{\Gamma}]^{I^{\boldp}} \cup I^{\boldq}
	\label{eq:1:prop:SM-INF}
\end{gather}
For this, note that
\begin{gather*}
\gr{I}{\Gamma'}
\equiv \gr{I}{\Gamma}
\end{gather*}
holds because $\gr{I}{\Gamma'}$ is the result of replacing each occurrence of~$p(\boldt)$ in $\gr{I}{\Gamma}$ by $\neg\neg p(\boldt)$
with ${p \in \boldq}$.
As a result,
$\gr{I}{\Gamma'}$
is the outcome of replacing each occurrence of~$X$ in $\gr{I}{\Gamma}$
(with ${X \in \{\top,\bot\}}$)
by~$\neg\neg X$.
Consequently,
equivalence~\eqref{eq:1:prop:SM-INF} is proven,
and we get that
any interpretation~$J$ satisfies that
$J^{\boldp\boldq} \models [\grp{\boldp\boldq}{I}{\Gamma_1'}]^{I^{\boldp\boldq}}$
iff
$J^{\boldp} \models [\gr{I}{\Gamma}]^{I^\boldp}$
and
${J^{\boldq} \models I^\boldq}$.
Note that
${J^{\boldq} \models I^\boldq}$
iff
${J^{\boldq} \supseteq I^\boldq}$.
Then,
\begin{align*}
&I \textrm{ is an INF-$\boldp\boldq$-stable model of } \Gamma_1'\\
\textrm{iff } &I^{\boldp\boldq} \textrm{ is a stable model of } \grp{\boldp\boldq}{I}{\Gamma_1'}\\
\textrm{iff } &I^{\boldp\boldq} \textrm{ is a model of } [\grp{\boldp\boldq}{I}{\Gamma_1'}]^{I^{\boldp\boldq}} \textrm{ and there is no model } J \subset I^{\boldp\boldq} \textrm{ of } [\grp{\boldp\boldq}{I}{\Gamma_1'}]^{I^{\boldp\boldq}}\\
\textrm{iff } &I^\boldp \textrm{ is a model of } [\gr{I}{\Gamma}]^{I^\boldp} \textrm{ and there is no model } J \subset I^{\boldp\boldq} \textrm{ of } [\gr{I}{\Gamma}]^{I^\boldp} \cup I^\boldq\\
\textrm{iff } &I^\boldp \textrm{ is a model of } [\gr{I}{\Gamma}]^{I^\boldp} \textrm{ and there is no model } J' \subset I^\boldp \textrm{ of } [\gr{I}{\Gamma}]^{I^\boldp}\\
\textrm{iff } &I^\boldp \textrm{ is a stable model of } [\gr{I}{\Gamma}]\\
\textrm{iff } &I \textrm{ is an INF-$\boldp$-stable model of } \Gamma. \tag*{\qedhere}
\end{align*}
\end{proof}

The following adaptation of Proposition 3 from \cite{lilusc19a}
to our notation is useful to prove the Main Lemma.

\begin{proposition}\label{prop:propositionalitation}
Any rule~$R$ and interpretation~$I$ satisfy $\gr{I}{\tau^*R} \equiv_s \tau R$.
\end{proposition}

\begin{proof}
By identifying the precomputed terms in $\prop{\tau^* \Pi}$ with their names in~$I$,
we get $\gr{I}{\tau^* \Pi} = \prop{\tau^* \Pi}$,
where~$\prop{\tau^* \Pi}$ is defined as in \cite[Section~5]{lilusc19a}.
\end{proof}

\begin{proof}[Proof of the Main Lemma]
Let $\Pi$~be a program,
let $\boldp$~be the list of all predicate symbols occurring in~$\Pi$ other than the comparison symbols,
and let $\I$~be a set of precomputed atoms.
By the choice of~$\boldp$,
we get that all predicate symbols in~$\Pi$ and none of the relations belong to~$\boldp$
and, therefore, $\I = \At{(\Iu)}$.
Then,
from Proposition~\ref{prop:SM-INF},
it follows that
$\Iu$~is a \SMstable{} model of~$\tau^* \Pi$
iff $\Iu$ is an \INFstable{} model of~$\tau^* \Pi$
iff $\I$ is a $\subseteq$-minimal model of $[\gr{\Iu}{\tau^* \Pi}]^{\I}$
iff $\I$ is a stable model of $[\gr{\Iu}{\tau^* \Pi}]$
iff $\I$ is a stable model of~$\tau \Pi$
(Proposition~\ref{prop:propositionalitation})
iff $\I$ is a stable model of~$\Pi$ (by definition).
\end{proof}

\subsection{Main Lemma for IO-Programs}
\label{sec:main-lemma.io-programs}

We need the following terminology to extend the Main Lemma to io-programs.
The models of formula $\exists \boldH\, (\SM_\boldp[F])^\boldh_\boldH$ are called the
\emph{{$\boldp$-stable} models with private symbols~$\boldh$}
of~$F$,
where $\boldH$~is a tuple of predicate variables of the same length as~$\boldh$
and $F^\boldh_\boldH$~is the result of replacing all occurrences of constants from~$\boldh$ by the corresponding variables from~$\boldH$.
For a set~$\Gamma$ of first-order sentences, the
\emph{{$\boldp$-stable} models with private symbols~$\boldh$} of~$\Gamma$ are the {$\boldp$-stable} models with private symbols~$\boldh$
of the conjunction of all formulas in~$\Gamma$ \cite{cafali20a}.
We usually omit parentheses and write just $\exists \boldH\, \SM_\boldp[F]^\boldh_\boldH$ instead of $\exists \boldH\, (\SM_\boldp[F])^\boldh_\boldH$.

\begin{mainlemmaio}\label{thm:main-lemma2}
Let $\Omega=(\Pi,\PH, \In, \Out)$ be an io-program,
let $\boldp$~be the list of all predicate
symbols occurring in~$\Pi$ other than the comparison and input symbols,
and let $\boldh$~be the list of all its private symbols.
A set~$\I$ of precomputed public atoms is an io-model of~$\Omega$ for an
input~$(\boldv,\boldi)$
iff $\I^\boldv$~is a \SMstable{} model with private symbols~$\boldh$ of~$\tau^*\Pi$ and $\I^{\it in} = \boldi$.
\end{mainlemmaio}

The following is a reformulation of the Splitting Theorem in \cite{felelipa09a} adapted to our notation, and it will be useful in proving the above result.
We adopt the following terminology.
An occurrence of a predicate symbol in a formula is called \emph{negated} if it belongs to a subformula of the form $F \to \bot$ and \emph{nonnegated} otherwise.
An occurrence of a predicate symbol in a formula is called \emph{positive} if the number of implications containing that occurrence in the antecedent is even.
It is called \emph{strictly positive} if that number is~$0$.
A \emph{rule} of a first-order formula~$F$ is a strictly positive occurrence of an implication in~$F$.
The \emph{dependency graph} of a formula is a directed graph that
\begin{itemize}
\item has all intensional predicate symbols as vertices and

\item has an edge from $p$ to~$q$ if, for some rule~$G \to H$ of~$F$,
formula~$G$ has a positive nonnegated occurrence of~$q$
and
$H$~has a strictly positive occurrence of~$p$.
\end{itemize}

A formula~$F$ is said to be negative on a tuple~$\boldp$ of predicate constants if members of~$\boldp$ have no strictly positive
occurrences in~$F$.

\begin{proposition}\label{prop:splitting.one-sorted}
Let $F$ and~$G$ be one-sorted first-order sentences and let $\boldp$ and~$\boldq$ be two disjoint tuples of distinct predicate symbols
such that
\begin{itemize}
\item each strongly connected component of the of the dependency graph of $F \wedge G$ is a subset either of $\boldp$ or~$\boldq$,
\item all occurrences in~$F$ of symbols from~$\boldq$ are negative, and
\item all occurrences in~$G$ of symbols from~$\boldp$ are negative.
\end{itemize}
Then,
$\SM_{\boldp\boldq}[F \wedge G]$ is equivalent to
$\SM_{\boldp}[F] \wedge \SM_{\boldq}[G]$.
\end{proposition}

This result can be straightforwardly lifted to the two-sorted language as follows.

\begin{proposition}\label{prop:splitting}
Let $F$ and~$G$ be two-sorted first-order sentences and let $\boldp$ and~$\boldq$ be two disjoint tuples of distinct predicate symbols
such that
\begin{itemize}
\item each strongly connected component of the dependency graph of~$F \wedge G$ is a subset either of $\boldp$ or~$\boldq$,
\item all occurrences in~$F$ of symbols from~$\boldq$ are negative, and
\item all occurrences in~$G$ of symbols from~$\boldp$ are negative.
\end{itemize}
Then,
$\SM_{\boldp\boldq}[F \wedge G]$ is equivalent to
$\SM_{\boldp}[F] \wedge \SM_{\boldq}[G]$.
\end{proposition}

\begin{proof}
Let~$I$ be any interpretation.
Then,
\begin{align*}
&I\models\SM_{\boldp\boldq}[F \wedge G]\\
\textrm{iff }
&\unsorted{I}\models\SM_{\boldp\boldq}[\unsorted{(F \wedge G)}]
	\tag*{(Lemma~\ref{lem:restricting.quantifiers.fo.sm})}\\
\textrm{iff }
&\unsorted{I}\models\SM_{\boldp\boldq}[\unsorted{F} \wedge \unsorted{G}]\\
\textrm{iff }
&\unsorted{I}\models\SM_{\boldp}[\unsorted{F}] \wedge \SM_{\boldq}[\unsorted{G}]
	\tag*{(Proposition~\ref{prop:splitting.one-sorted})}\\
\textrm{iff }
&I \models \SM_{\boldp}[F] \wedge \SM_{\boldq}[G].
	\tag*{(Lemma~\ref{lem:restricting.quantifiers.fo.sm})}
\end{align*}
Note that the dependency graphs of $F \wedge G$ and $\unsorted{F} \wedge \unsorted{G}$ are the same.
\end{proof}

\begin{lemma}\label{thm:main-lemma2.aux}
Let $\Omega=(\Pi,\PH, \In, \Out)$ be an io-program and
let $\boldp$~be the list of all predicate
symbols occurring in~$\Pi$ other than the comparison and input symbols.
A set~$\I$ of precomputed public atoms is a stable model of $\Omega(\boldv,\boldi)$
iff $\I^\boldv$~is a model of $\SM_{\boldp}[\tau^*\Pi]$ and $\I^{\mathit{in}} = \boldi$.
\end{lemma}

\begin{proof}
Recall that, from the Main Lemma stated in Section~\ref{sec:main.lemma},
we get that
$\I$~is a stable model of $\Omega(\boldv,\boldi)$
iff $\Iu$~is a $\boldp\boldq$-stable model of $\tau^*(\Omega(\boldv,\boldi))$
iff
$\Iu$~is a model of $\SM_{\boldp\boldq}[\tau^*(\Omega(\boldv,\boldi))]$,
where $\boldq$~is the list of all input symbols.
Let us denote by~$\Omega(\boldv)$ the set of rules obtained from the rules of~$\Omega$
by substituting the precomputed terms~$v(c)$ for all occurrences of all placeholders~$c$.
Then,
$\tau^*(\Omega(\boldv,\boldi)) = \tau^*(\Omega(\boldv) \cup \boldi) = \tau^*(\Omega(\boldv)) \cup \tau^*(\boldi) \equiv_s \tau^*(\Omega(\boldv)) \cup \boldi$.
Furthermore, since there are no occurrences of predicate symbols in~$\boldq$ in the heads of the rules of~$\Omega(\boldv)$
nor of any predicate symbol in~$\boldp$ in the head of the rules in~$\boldi$,
we get that each strongly connected component is a subset either of $\boldp$ or~$\boldq$.
From Proposition~\ref{prop:splitting},
this implies that
\begin{align*}
&\Iu \textrm{ is a model of } \SM_{\boldp\boldq}[\tau^*(\Omega(\boldv,\boldi))]\\
\textrm{iff }
&\Iu \textrm{ is a model of } \SM_{\boldp}[\tau^*(\Omega(\boldv))]
\textrm{ and }
\Iu \textrm{ is a model of } \SM_{\boldq}[\boldi]\\
\textrm{iff }
&\I^\boldv \textrm{ is a model of } \SM_{\boldp}[\tau^*\Pi]
\textrm{ and }
\I^{\mathit{in}} = \boldi.
\end{align*}
For the second equivalence,
note that $\I^\boldv$~is identical to~$\Ju$, except that it interprets each placeholder~$c$ as~$\boldv(c)$ and that~$\Omega(\boldv)$ is the result of replacing each placeholder~$c$ by~$\boldv(c)$.
\end{proof}

\begin{proof}[Proof of the Main Lemma for IO-Programs]
From left to right.
Assume that
$\I$~is an io-model of~$\Omega$ for input~$(\boldv,\boldi)$.
Let us show that $\I^\boldv$~is a \SMstable{} model with private symbols~$\boldh$ of $\tau^*\Pi$ and $\I^{\mathit{in}} = \boldi$.
By definition,
the assumption implies that
there is some stable model~$\J$
of~$\Omega(\boldv,\boldi)$
such that $\I$~is the set of all public atoms of~$\J$.
From Lemma~\ref{thm:main-lemma2.aux},
this implies that
$\J^\boldv$~is a model of $\SM_{\boldp}[\tau^*\Pi]$ and $\J^{\mathit{in}} = \boldi$
and, thus, that
$\I^\boldv$~is a model of $\exists\boldH\,\SM_{\boldp}[(\tau^*\Pi)]^\boldh_\boldH$
and $\I^{\mathit{in}} = \boldi$.
For this last step, recall that $\I$ and~$\J$ agree on all public predicates.
By definition,
this means that $\I^\boldv$~is a \SMstable{} model with private symbols~$\boldh$ of~$\tau^*\Pi$ and $\I^{\mathit{in}} = \boldi$.

From right to left.
Assume that $\I^\boldv$~is a \SMstable{} model with private symbols~$\boldh$ of~$\tau^*\Pi$.
Let us show that $\I$~is an io-model of~$\Omega$ for an
input $(\boldv,\boldi)$.
By definition,
the assumption implies that
$\I^\boldv$~is a model of $\exists\boldH\,\SM_{\boldp}[(\tau^*\Pi)]^\boldh_\boldH$.
This implies that there is some model~$J$ of $\SM_{\boldp}[(\tau^*\Pi)]$
such that $\I^\boldv$ and~$J$ agree on the interpretation of all public predicates.
Let $\J$~be the set of precomputed atoms satisfied by~$J$.
Then, there are no occurrences of placeholders in~$\J$ and, thus,
we get that $\J^\boldv = \Ju = J$
and, thus, also that
$\J^\boldv$~is a stable model of $\SM_{\boldp}[(\tau^*\Pi)]$.
Recall that we also have $\I^{\mathit{in}} = \boldi$
and, since $\I$ and~$\J$ contain the same public atoms,
we get that $\J^{\mathit{in}} = \boldi$.
From Lemma~\ref{thm:main-lemma2.aux},
these two facts together imply that
$\J$~is a stable model of $\Omega(\boldv,\boldi)$
and, therefore, that
$\I$~is an io-model of~$\Omega$ for an
input~$(\boldv,\boldi)$.
\end{proof}
 \subsection{Theorem~\ref{th1}}
\label{sec:proof:thm:completion.correspondence}

In order to prove Theorem~\ref{th1},
we need the notions of Clark normal form,
completion, and tight theories~\cite[Section~6]{feleli11a}.
We adapt these notions to a two-sorted language here.
A theory---one-sorted or two-sorted---is in \emph{Clark normal form} relative to a list~$\boldp$ of intensional predicates if it contains exactly one sentence of the form
\begin{gather}
\forall V_1 \ldots V_n \, ( G \to p(V_1, \ldots, V_n) )
	\label{eq:clark}
\end{gather}
for each intensional predicate symbol~$p/n$ in~$\boldp$,
where $G$~is a formula
and $V_1,\ldots,V_n$ are distinct program variables.
The \emph{completion} of a theory~$\Gamma$ in Clark normal form,
denoted by $\COMP_\boldp[\Gamma]$, is obtained by replacing each implication~$\to$ by an equivalence~$\leftrightarrow$ in all sentences of form~\eqref{eq:clark}.
The following is a special case of Theorem~10 in~\cite{feleli11a} adapted to our notation.

\begin{proposition}\label{prop:completion.one-sorted.if}
For any one-sorted sentence~$F$ in Clark normal form and list of predicates~$\boldp$, the implication
\begin{gather*}
\SM_\boldp[F] \to \COMP_\boldp[F]
\end{gather*}
is satisfied by all one-sorted interpretations.
\end{proposition}

Then,
we can easily extend this result to two-sorted interpretations as follows.

\begin{proposition}\label{prop:completion.if}
For any two-sorted sentence~$F$ in Clark normal form, list of predicates~$\boldp$,
and two-sorted interpretation~$\I$,
if $\I$~satisfies $\SM_\boldp[F]$,
then it also satisfies $\COMP_\boldp[F]$.
\end{proposition}

\begin{proof}
Let $I$~be any two-sorted interpretation.
From Lemma~\ref{lem:restricting.quantifiers.fo},
we get that
${I \models \COMP_\boldp[F]}$ iff $\unsorted{I} \models \unsorted{(\COMP_\boldp[F])}$.
Furthermore,
we can see that
$\unsorted{(\COMP_\boldp[F])}=\COMP_\boldp[\unsorted{F}]$,
and thus,
we get
\begin{gather*}
I \models \COMP_\boldp[F] \textrm{ iff } \unsorted{I} \models \COMP_\boldp[\unsorted{F}].
\end{gather*}
Similarly, from Lemma~\ref{lem:restricting.quantifiers.fo.sm},
we get
\begin{gather*}
I \models \SM_\boldp[\Gamma] \textrm{ iff } \unsorted{I} \models \SM_\boldp[\unsorted{\Gamma}].
\end{gather*}
Finally, from Proposition~\ref{prop:completion.one-sorted.if},
we get
\begin{gather*}
\unsorted{I} \models \SM_\boldp[\unsorted{\Gamma}]
\textrm{ implies }
\unsorted{I} \models \COMP_\boldp[\unsorted{\Gamma}].
\end{gather*}
Consequently, the result holds.
\end{proof}

Let us introduce the \emph{Clark form} of a program without input and output.
The \emph{Clark definition} of~$p/n$ in~$\Pi$ is a formula of the form
\begin{gather}
\forall V_1 \ldots V_n \,
	\left(
	\bigvee_{i=1}^k \exists \boldU_i \, F_i
	\rightarrow
	p(V_1, \ldots ,V_n)
	\right),
	\label{eq:definition.p}
\end{gather}
where each $F_i$~is the formula representation of rule~$R_i$ and rules~$R_1,\ldots,R_k$ constitute the definition of~$p/n$ in~$\Pi$.
By~$\Cdf{\Pi}$, we denote the theory containing the Clark definitions of all predicate symbols.
We also define $\Clark{\Pi} \eqdef \Cdf{\Pi} \cup \Pi_C$,
where $\Pi_C$~is the set containing the formula representation of all constraints in~$\Pi$.

About first-order formulas $F$ and~$G$, it is said that $F$~is \emph{strongly equivalent} to~$G$ if, for any formula~$H$, any occurrence of~$F$ in~$H$, and any list~$\boldp$ of distinct predicate symbols, $\SM_\boldp[H]$ is equivalent to $\SM_\boldp[H']$, where $H'$~is obtained from~$H$ by replacing the occurrence of~$F$ by~$G$.
About finite first-order theories $\Gamma$ and~$\Gamma'$, we say that $\Gamma$~is \emph{strongly equivalent} to~$\Gamma'$ when the conjunction of all sentences in~$\Gamma$ is strongly equivalent to the conjunction of all sentences in~$\Gamma'$.
First-order theory~$\Gamma$ is strongly equivalent to~$\Gamma'$ iff $\Gamma$~is equivalent to~$\Gamma'$ in quantified equilibrium logic \cite[Theorem~8]{feleli11a}.
Therefore, $\Gamma$~is strongly equivalent to~$\Gamma'$ if $\Gamma$ is equivalent to~$\Gamma'$ in intuitionistic logic.

\begin{lemma}\label{lem:df.equivalence}
Let $\Pi$~be a program without constraints.
Then, $\tau^*\Pi$ is strongly equivalent to $\Cdf{\Pi}$.
\end{lemma}

\begin{proof}
By definition, $\tau^*\Pi$~contains a formula of the form
\begin{gather}
\forall V_1 \ldots V_n \boldU_i \,
	(
	F_i
		\rightarrow
	p(V_1,\ldots,V_n)
	)
		\label{eq:1:lem:df.equivalence}
\end{gather}
for each rule~$R_i$ in~$\Pi$.
Note that~\eqref{eq:1:lem:df.equivalence}
is strongly equivalent to
\begin{gather*}
\forall V_1 \ldots V_n \,
	(
	\exists \boldU_i \, F_i
		\rightarrow
	p(V_1,\ldots,V_n)
	).
\end{gather*}
Furthermore,
since $\Pi$~is finite,
it follows that $\tau^*\Pi$~is finite too and, therefore,
we get that $\tau^*\Pi$~is strongly equivalent to~$\Gamma$,
where~$\Gamma$ is the theory containing a formula of the form
\begin{gather}
\bigwedge_{i = 1}^k \forall V_1 \ldots V_n \,
	(
	\exists \boldU_i \, F_i
		\rightarrow
	p(V_1,\ldots,V_n)
	)
	\label{eq:2b:lem:df.equivalence}
\end{gather}
for each predicate symbol~$p/n$.
Finally, since \eqref{eq:definition.p} and~\eqref{eq:2b:lem:df.equivalence} are strongly equivalent, we get that~$\tau^*\Pi$ and~$\Cdf{\Pi}$ are also strongly equivalent.
\end{proof}

\begin{lemma}\label{lem:clark.equivalence}
For any program~$\Pi$,
$\tau^*\Pi$ is strongly equivalent to $\Clark{\Pi}$.
\end{lemma}

\begin{proof}
Let $\Pi_1,\Pi_2$ be a partition of~$\Pi$ such that $\Pi_2$~contains all constraints and $\Pi_1$ all the remaining rules.
Then, $\tau^*\Pi= \tau^*(\Pi_1 \cup \Pi_2) = \tau^*\Pi_1 \cup \tau^*\Pi_2 = \tau^*\Pi_1 \cup \Pi_C$.
Now, the result follows directly from Lemma~\ref{lem:df.equivalence}.
\end{proof}

\begin{proof}[Proof of Theorem~\ref{th1}]
From the Main Lemma for IO-Programs,
it follows that $\I$~is an io-model of~$\Omega$ for an
input~$(\boldv,\boldi)$
iff $\I^\boldv$~is a \SMstable{} model with private symbols~$\boldh$ of $\tau^*\Pi$ and $\I^{\it in} = \boldi$.
Furthermore,
\begin{align*}
&\I^\boldv \textrm{ is a $\SMstable$ model with private symbols } \boldh \textrm{ of } \tau^*\Pi\\
\textrm{iff } &\I^\boldv \textrm{ is a model of } \exists \boldH \, \SM_\boldp[\tau^*\Pi]^\boldh_\boldH\\
\textrm{iff } &\I^\boldv \textrm{ is a model of } \exists \boldH \, \SM_\boldp[\Clark{\Pi}]^\boldh_\boldH
	\tag*{(Lemma~\ref{lem:clark.equivalence})}\\
\textrm{ implies that } &\I^\boldv \textrm{ is a model of } \exists \boldH \, \COMP_\boldp[\Clark{\Pi}]^\boldh_\boldH
	\tag*{(Proposition~\ref{prop:completion.if})}\\
\textrm{iff } &\I^\boldv \textrm{ is a model of } \exists \boldH \, \COMP_\boldp[\tau^*\Pi]^\boldh_\boldH\\
\textrm{iff } &\I^\boldv \textrm{ is a model of } \COMP[\Omega]. \tag*{\qedhere}
\end{align*}
\end{proof}

\subsection{Theorem~\ref{th2}}
\label{sec:proof:thm:completion.correspondence2}

To prove Theorem~\ref{th2}, we need the following terminology.
An occurrence of a predicate symbol in a formula is called \emph{negated} if it belongs to a subformula of the form $F \to \bot$ and \emph{nonnegated} otherwise.
An occurrence of a predicate symbol in a formula is called \emph{positive} if the number of implications containing that occurrence in the antecedent is even.
The \emph{dependency graph} of a theory in Clark normal form is a directed graph that
\begin{itemize}
\item has all intensional predicate symbols as vertices and
\item has an edge from $p$ to~$q$ if $q$ has a positive nonnegated occurrence in~$G$
for some sentence of form~\eqref{eq:clark}.
\end{itemize}
A theory is \emph{tight} if its predicate dependency graph is acyclic.
The following is a reformulation of Theorem~11 in~\cite{feleli11a} adapted to our notation.

\begin{proposition}\label{prop:completion.one-sorted}
For any finite, tight, one-sorted theory~$\Gamma$ in Clark normal form,~$\SM_\boldp[\Gamma]$ is equivalent to~$\COMP_\boldp[\Gamma]$.
\end{proposition}

The following lifts this result to the case of two sorts.

\begin{proposition}\label{prop:completion}
For any finite, tight, two-sorted theory~$\Gamma$ in Clark normal form, $\SM_\boldp[\Gamma]$ is equivalent to $\COMP_\boldp[\Gamma]$.
\end{proposition}

\begin{proof}
Let $I$~be any two-sorted interpretation.
From Lemma~\ref{lem:restricting.quantifiers.fo},
we get that
$I \models \COMP_\boldp[\Gamma]$ iff $\unsorted{I} \models \unsorted{(\COMP_\boldp[\Gamma])}$.
Furthermore,
it is easy to see that
$\unsorted{(\COMP_\boldp[\Gamma])}=\COMP_\boldp[\unsorted{\Gamma}]$,
and thus,
we get
\begin{gather*}
I \models \COMP_\boldp[\Gamma] \textrm{ iff } \unsorted{I} \models \COMP_\boldp[\unsorted{\Gamma}].
\end{gather*}
Similarly, from Lemma~\ref{lem:restricting.quantifiers.fo.sm},
we get
\begin{gather*}
I \models \SM_\boldp[\Gamma] \textrm{ iff } \unsorted{I} \models \SM_\boldp[\unsorted{\Gamma}].
\end{gather*}
Finally, from Proposition~\ref{prop:completion.one-sorted},
we get
\begin{gather*}
\unsorted{I} \models \COMP_\boldp[\unsorted{\Gamma}] \textrm{ iff } \unsorted{I} \models \SM_\boldp[\unsorted{\Gamma}].
\end{gather*}
Consequently, the result holds.
\end{proof}

\begin{proof}[Proof of Theorem~\ref{th2}]
From the Main Lemma for IO-Programs, it follows that
$\I$~is an io-model of~$\Omega$ for an
input~$(\boldv,\boldi)$
iff $\I^\boldv$~is a \SMstable{} model with private symbols~$\boldh$ of $\tau^*\Pi$ and $\I^{\it in} = \boldi$.
Furthermore,
\begin{align*}
&\I^\boldv \textrm{ is a \SMstable{} model with private symbols } \boldh \text{ of}~\tau^*\Pi\\
\textrm{iff } &\I^\boldv \textrm{ is a model of } \exists \boldH\, \SM_\boldp[\tau^*\Pi]^\boldh_\boldH\\
\textrm{iff } &\I^\boldv \textrm{ is a model of } \exists \boldH\, \SM_\boldp[\Clark{\Pi}]^\boldh_\boldH
	\tag*{(Lemma~\ref{lem:clark.equivalence})}\\
\textrm{iff } &\I^\boldv \textrm{ is a model of } \exists \boldH\, \COMP_\boldp[\Clark{\Pi}]^\boldh_\boldH
	\tag*{(Proposition~\ref{prop:completion})}\\
\textrm{iff } &\I^\boldv \textrm{ is a model of } \exists \boldH\, \COMP_\boldp[\tau^*\Pi]^\boldh_\boldH\\
\textrm{iff } &\I^\boldv \textrm{ is a model of } \COMP[\Omega].
\end{align*}
Recall that $\Omega$~is tight,
and this implies that $\Clark{\tau^*\Pi}$ is also tight.
Note that $\Clark{\tau^*\Pi}$ 
contains a formula of form~\eqref{eq:definition.p} for every predicate symbol~$p/n$ and that the antecedent of this formula is a disjunction of the formula representations of the bodies of all rules defining~$p/n$.
Therefore,
the dependency graph of~$\Clark{\tau^*\Pi}$ is identical to the dependency graph of~$\Omega$ with the exception of the addition of nonpositive edges corresponding to choice rules.
\end{proof}

\subsection{Theorem~\ref{th3}}

A \emph{predicate expression} is a lambda expression of the form
\begin{gather}
\lambda X_1 \ldots X_n \, F(X_1,\ldots,X_n),
\label{eq:pred.exp}
\end{gather}
where $F(X_1,\ldots,X_n)$ is a formula and $X_1, \ldots, X_n$ are object variables.
This formula may have free variables other
than $X_1, \ldots, X_n$, called the \emph{parameters} of~\eqref{eq:pred.exp}.
If $E$ is~\eqref{eq:pred.exp} and $t_1,\ldots,t_n$ are terms,
then $E(t_1,\ldots,t_n)$ stands for the formula~$F(t_1,\ldots,t_n)$.
If $G(P)$~is a formula containing a predicate constant or variable~$P$
and $E$~is a predicate expression of the same arity as $P$,
then $G(E)$~stands for the result of replacing
each atom $P(t_1,\ldots,t_n)$ in~$G(P)$ by $E(t_1,\ldots,t_n)$.
For any predicate
expression~$E$, the formulas
\begin{gather*}
\forall P\, G(P) \to A(E) \quad\textrm{and}\quad G(E) \to \exists P\, A(P)
\end{gather*}
are theorems of second-order logic.

\begin{lemma}\label{lem:induction:thm3}
Let $\boldP = P_1,\ldots,P_l$ be a list of predicate variables and
let $\boldP_i = P_1,\ldots,P_i$ be a prefix of~$\boldP$.
Let $F_1(\boldP_1),\ldots,F_l(\boldP_l)$ be formulas such that $\boldP_i$~contains all free predicate variables occurring in~$F_i$.
Let $F$ and~$G$, respectively, be the following two formulas:
\begin{gather}
\exists \boldP \, (F_1(\boldP_1) \wedge \cdots \wedge F_l(\boldP_l) \wedge F'(\boldP)),
	\label{eq:1:lem:induction:thm3}
\\
\forall \boldP \, (F_1(\boldP_1) \wedge \cdots \wedge F_l(\boldP_l) \to F'(\boldP)).
	\label{eq:2:lem:induction:thm3}
\end{gather}
Then, $F \equiv G$.
\end{lemma}

\begin{proof}
If $l=0$, then $\boldP$~is the empty tuple, and thus,
both \eqref{eq:1:lem:induction:thm3} and~\eqref{eq:2:lem:induction:thm3} stand just for the formula~$F'(\boldP)$, so the result holds.
Otherwise,
we proceed by induction.
Note that \eqref{eq:1:lem:induction:thm3} and~\eqref{eq:2:lem:induction:thm3}
are, respectively, equivalent to
\begin{gather}
\exists P_1 \, (F_1(P_1) \wedge \exists \boldP_l' \, (F_2(P_1,\boldP_2') \wedge \cdots \wedge F_l(P_1,\boldP_l') \wedge F'(P_1,\boldP_l') )),
	\label{eq:thm:private.recursion.2}
\\
\forall P_1 \, (F_1(P_1) \to \forall \boldP_l' \, (F_2(P_1,\boldP_2') \wedge \cdots \wedge F_l(P_1,\boldP_l') \to F'(P_1,\boldP_l'))).
	\label{eq:thm:private.recursion.2b}
\end{gather}
where $\boldP_i' = P_2,\ldots,P_i$.
That is, $\boldP_i = P_1,\boldP_i'$.
Then, by induction hypothesis, we get that the following two formulas are equivalent:
\begin{gather*}
\exists \boldP_l' \, (F_2(P_1,\boldP_2') \wedge \cdots \wedge F_l(P_1,\boldP_l') \wedge F'(P_1,\boldP_l') ),
\\
\forall \boldP_l' \, (F_2(P_1,\boldP_2') \wedge \cdots \wedge F_l(P_1,\boldP_l') \to F'(P_1,\boldP_l')).
\end{gather*}
Therefore, \eqref{eq:thm:private.recursion.2b}
is equivalent to
\begin{gather}
\forall P_1 \, (F_1(P_1) \to \exists \boldP_l' \, (F_2(P_1,\boldP_2') \wedge \cdots \wedge F_l(P_1,\boldP_l') \wedge F'(P_1,\boldP_l') )).
	\label{eq:thm:private.recursion.2c}
\end{gather}
Hence,
it only remains to be shown that \eqref{eq:thm:private.recursion.2}
and~\eqref{eq:thm:private.recursion.2c}
are equivalent.
Let $E$~be the predicate expression
$\lambda X_1 \ldots X_n\, G(X_1, \ldots, X_n)$
such that
$H(E) = F_1(P_1)$, with $H(Q)$~being the following formula:
\begin{gather*}
\forall V_1 \ldots V_n \, ( P_1(V_1,\ldots,V_n) \leftrightarrow Q(V_1,\ldots,V_n) ).
\end{gather*}
Then,
\begin{align*}
\eqref{eq:thm:private.recursion.2}
&\Leftrightarrow
\exists P_1 \, (F_1(P_1) \wedge \exists \boldP_2 \, (F_2(E,\boldP_2') \wedge \cdots \wedge F_l(E,\boldP_l') \wedge F'(E,\boldP_l') ))
\\
&\Leftrightarrow
\exists P_1 \, F_1(P_1) \wedge \exists \boldP_2 \, (F_2(E,\boldP_2') \wedge \cdots \wedge F_l(E,\boldP_l') \wedge F'(E,\boldP_l') )
\\
&\Leftrightarrow
\top \wedge \exists \boldP_2 \, (F_2(E,\boldP_2') \wedge \cdots \wedge F_l(E,\boldP_l') \wedge F'(E,\boldP_l') )
\\
&\Leftrightarrow \exists \boldP_2 \, (F_2(E,\boldP_2') \wedge \cdots \wedge F_l(E,\boldP_l') \wedge F'(E,\boldP_l' ))
\\
&\Leftrightarrow \forall P_1\, (F_1(P_1) \to \exists \boldP_2 \, (F_2(E,\boldP_2') \wedge \cdots \wedge F_l(E,\boldP_l') \wedge F'(E,\boldP_l')))
\\
&\Leftrightarrow
\eqref{eq:thm:private.recursion.2c},
\end{align*}
and the result holds.
For the second-to-last equivalence, note that $F_1(P_1)$ is satisfiable and that $P_1$~does not occur on the right-hand side of the implication.
\end{proof}

\begin{proof}[Proof of Theorem~\ref{th3}]
Recall that io-program~$\Omega$ uses private recursion if
\begin{itemize}
	\item its predicate dependency graph has a cycle such that every vertex in it is a private symbol or
	\item it includes a choice rule with a private symbol in the head.
\end{itemize}
This implies that, for a program that does not use private recursion, there is a private predicate symbol that does not depend on any other private predicate symbol.
Let us assume without loss of generality that this is the predicate symbol $p_1/n_1$.
Then,
there is a predicate symbol that does not depend on any other private predicate symbol other than $p_1/n_1$, which we assume to be the predicate symbol $p_2/n_2$,
and so on.
Therefore,
we have an order on the private symbols $p_1/n_1,\ldots,p_l/n_l$
such that each predicate symbol $p_i/n_i$ only depends on other predicate symbols that precede them in this order.
Then, the completed definition $F_i(\boldP_i)$ of any private predicate symbol $p_i/n_i$ can be written as
\begin{align*}
	\forall V_1 \ldots V_{n_i}  \, (P_i(V_1, \dots, V_{n_i}) \lrar G_i(\boldP_{i-1})),
\end{align*}
where $\boldP_{i-1} = P_1,\ldots,P_{i-1}$ contains all free predicate variables in~$G_i$
and where we assume that ${G_1(\boldP_0)}$ is a first-order formula.
Then, \eqref{eq:private.recursion.existencial} and~\eqref{eq:private.recursion.universal} can be, respectively, rewritten as \eqref{eq:1:lem:induction:thm3} and~\eqref{eq:2:lem:induction:thm3}.
The result follows then directly from Lemma~\ref{lem:induction:thm3}.
\end{proof}
 
\end{document}